\documentclass[letterpaper,11pt,]{article}

\usepackage[utf8]{inputenc}

\usepackage{amsmath, amssymb}
\usepackage{algpseudocode}
\usepackage{etoolbox}

\usepackage{etex}
\usepackage{verbatim}
\usepackage{xspace,enumerate}
\usepackage[dvipsnames]{xcolor}
\usepackage[T1]{fontenc}
\usepackage[full]{textcomp}
\usepackage[american]{babel}
\usepackage{mathtools}
\usepackage{amsthm}
\usepackage[
letterpaper,
top=1in,
bottom=1in,
left=1in,
right=1in]{geometry}
\usepackage{newpxtext} 
\usepackage{textcomp} 
\usepackage[varg,bigdelims]{newpxmath}
\usepackage[scr=rsfso]{mathalfa}
\usepackage{bm} 
\let\mathbb\varmathbb
\usepackage{microtype}
\usepackage[pagebackref,colorlinks=true,urlcolor=blue,linkcolor=blue,citecolor=OliveGreen]{hyperref}
\usepackage[capitalise,nameinlink]{cleveref}

\let\definitionref\cref
\let\lemmaref\cref

\crefname{lemma}{Lemma}{Lemmas}
\crefname{fact}{Fact}{Facts}
\crefname{theorem}{Theorem}{Theorems}
\crefname{corollary}{Corollary}{Corollaries}
\crefname{claim}{Claim}{Claims}
\crefname{example}{Example}{Examples}
\crefname{algorithm}{Algorithm}{Algorithms}
\crefname{problem}{Problem}{Problems}
\crefname{definition}{Definition}{Definitions}
\crefname{exercise}{Exercise}{Exercises}
\crefname{condition}{Condition}{Conditions}
\crefname{property}{Property}{Properties}

\usepackage{newfloat}
\usepackage{array}
\usepackage{subfig}
\usepackage{bbm}
\usepackage{xparse}
\usepackage{amsthm} 

\newtheorem{theorem}{Theorem}[section]
\newtheorem*{theorem*}{Theorem}
\newtheorem{lemma}[theorem]{Lemma}
\newtheorem*{lemma*}{Lemma}
\newtheorem{fact}[theorem]{Fact}
\newtheorem*{fact*}{Fact}
\newtheorem{proposition}[theorem]{Proposition}
\newtheorem*{proposition*}{Proposition}
\newtheorem{corollary}[theorem]{Corollary}
\newtheorem*{corollary*}{Corollary}

\newtheorem*{hypothesis*}{Hypothesis}

\newtheorem*{conjecture*}{Conjecture}
\theoremstyle{definition}
\newtheorem{definition}[theorem]{Definition}
\newtheorem*{definition*}{Definition}

\newtheorem*{construction*}{Construction}

\newtheorem*{example*}{Example}

\newtheorem*{question*}{Question}
\newtheorem{algorithm}[theorem]{Algorithm}
\newtheorem*{algorithm*}{Algorithm}

\newtheorem*{assumption*}{Assumption}

\newtheorem*{problem*}{Problem}

\newtheorem*{openquestion*}{Open Question}
\theoremstyle{remark}

\newtheorem*{claim*}{Claim}

\newtheorem*{remark*}{Remark}

\newtheorem*{observation*}{Observation}

\usepackage{times}

\usepackage{mathtools}

\usepackage{bm} 

\let\originalleft\left
\let\originalright\right
\renewcommand{\left}{\mathopen{}\mathclose\bgroup\originalleft}
\renewcommand{\right}{\aftergroup\egroup\originalright}

\usepackage{turnstile}
\usepackage{mdframed}
\usepackage{tikz}
\usetikzlibrary{positioning}
\usepackage{caption}
\usepackage{newfloat}
\usepackage{array}
\usepackage{bbm}
\usepackage{xparse}
\makeatletter
\let\latexparagraph\paragraph
\RenewDocumentCommand{\paragraph}{som}{%
	\IfBooleanTF{#1}
	{\latexparagraph*{#3}}
	{\IfNoValueTF{#2}
		{\latexparagraph{\maybe@addperiod{#3}}}
		{\latexparagraph[#2]{\maybe@addperiod{#3}}}%
	}%
}
\newcommand{\maybe@addperiod}[1]{%
	#1\@addpunct{.}%
}
\makeatother

\usepackage{boxedminipage}
\newcommand{\paren}[1]{(#1)}
\newcommand{\Paren}[1]{\left(#1\right)}

\newcommand{\brac}[1]{[#1]}
\newcommand{\Brac}[1]{\left[#1\right]}

\newcommand{\abs}[1]{\lvert#1\rvert}
\newcommand{\Abs}[1]{\left\lvert#1\right\rvert}

\newcommand{\card}[1]{\lvert#1\rvert}
\newcommand{\Card}[1]{\left\lvert#1\right\rvert}

\newcommand{\set}[1]{\{#1\}}
\newcommand{\Set}[1]{\left\{#1\right\}}

\newcommand{\norm}[1]{\lVert#1\rVert}
\newcommand{\Norm}[1]{\left\lVert#1\right\rVert}

\newcommand{\normi}[1]{\norm{#1}_\infty}

\newcommand{\iprod}[1]{\langle#1\rangle}
\newcommand{\Iprod}[1]{\left\langle#1\right\rangle}

\newcommand{\Esymb}{\mathbb{E}}
\newcommand{\Psymb}{\mathbb{P}}

\newcommand{\Covsymb}{\mathrm{Cov}}
\DeclareMathOperator*{\E}{\Esymb}

\DeclareMathOperator*{\ProbOp}{\Psymb}
\DeclareMathOperator*{\Cov}{\Covsymb}
\renewcommand{\Pr}{\ProbOp}

\newcommand{\suchthat}{\;\middle\vert\;}

\newcommand{\pE}{\tilde{\mathbb{E}}}
\newcommand{\mper}{\,.}

\newcommand\bdot\bullet

\DeclareMathOperator{\Tr}{Tr}

\DeclareMathOperator{\poly}{poly}

\DeclareMathOperator{\polylog}{polylog}

\DeclareMathOperator{\sign}{sign}
\DeclareMathOperator{\conv}{conv}
\DeclareMathOperator{\Conv}{Conv}

\newcommand{\iid}{i.i.d.\xspace}

\newcommand{\N}{\mathbb N}
\newcommand{\R}{\mathbb R}

\newcommand{\cB}{\mathcal B}
\newcommand{\cC}{\mathcal C}
\newcommand{\cD}{\mathcal D}

\newcommand{\cF}{\mathcal F}

\newcommand{\cN}{\mathcal N}

\newcommand{\cP}{\mathcal P}
\newcommand{\cQ}{\mathcal Q}

\newcommand{\cS}{\mathcal S}

\newcommand{\cV}{\mathcal V}
\newcommand{\cW}{\mathcal W}

\newcommand{\cY}{\mathcal Y}

\renewcommand{\leq}{\leqslant}
\renewcommand{\le}{\leqslant}

\renewcommand{\ge}{\geqslant}
\let\epsilon=\varepsilon
\numberwithin{equation}{section}
\newcommand\MYcurrentlabel{xxx}
\newcommand{\MYstore}[2]{%
  \global\expandafter \def \csname MYMEMORY #1 \endcsname{#2}%
}
\newcommand{\MYload}[1]{%
  \csname MYMEMORY #1 \endcsname%
}
\newcommand{\MYnewlabel}[1]{%
  \renewcommand\MYcurrentlabel{#1}%
  \MYoldlabel{#1}%
}
\newcommand{\MYdummylabel}[1]{}
\newcommand{\torestate}[1]{%
  \let\MYoldlabel\label%
  \let\label\MYnewlabel%
  #1%
  \MYstore{\MYcurrentlabel}{#1}%
  \let\label\MYoldlabel%
}
\newcommand{\restatetheorem}[1]{%
  \let\MYoldlabel\label
  \let\label\MYdummylabel
  \begin{theorem*}[Restatement of \cref{#1}]
    \MYload{#1}
  \end{theorem*}
  \let\label\MYoldlabel
}
\newcommand{\restatelemma}[1]{%
  \let\MYoldlabel\label
  \let\label\MYdummylabel
  \begin{lemma*}[Restatement of \cref{#1}]
    \MYload{#1}
  \end{lemma*}
  \let\label\MYoldlabel
}
\newcommand{\restateprop}[1]{%
  \let\MYoldlabel\label
  \let\label\MYdummylabel
  \begin{proposition*}[Restatement of \cref{#1}]
    \MYload{#1}
  \end{proposition*}
  \let\label\MYoldlabel
}
\newcommand{\restatefact}[1]{%
  \let\MYoldlabel\label
  \let\label\MYdummylabel
  \begin{fact*}[Restatement of \cref{#1}]
    \MYload{#1}
  \end{fact*}
  \let\label\MYoldlabel
}
\newcommand{\restate}[1]{%
  \let\MYoldlabel\label
  \let\label\MYdummylabel
  \MYload{#1}
  \let\label\MYoldlabel
}

\newcommand{\e}{\epsilon}

\allowdisplaybreaks
\newcommand*{\Id}{\mathrm{Id}}

\newcommand*{\normf}[1]{\norm{#1}_{\mathrm{F}}}
\newcommand*{\Normf}[1]{\Norm{#1}_{\mathrm{F}}}

\newcommand{\ind}[1]{\mathbf{1}_{\Brac{#1}}}

\newcommand*{\transpose}[1]{{#1}{}^{\mkern-1.5mu\mathsf{T}}}

\newcommand{\simiid}{\stackrel{\text{iid}}\sim}

\newcommand{\CS}{Cauchy--Schwarz }

\newcommand{\effrank}{\mathrm{erk}}
\newcommand{\spsign}{\textrm{spsign}}

\usepackage{times}

\title{Nearly Optimal Robust Covariance and Scatter Matrix Estimation Beyond Gaussians}

\author{Gleb Novikov\thanks{Lucerne School of Computer Science and Information Technology}}

\begin{document}

\clearpage\maketitle
\thispagestyle{empty}

\begin{abstract}%
We study the problem of \emph{computationally efficient} robust estimation of the covariance/scatter matrix of elliptical distributions---that is, affine transformations of spherically symmetric distributions---under the \emph{strong contamination model} in the high-dimensional regime $d \gtrsim 1/\varepsilon^2$, where $d$ is the dimension and $\varepsilon$ is the fraction of adversarial corruptions. We show that the structure inherent to elliptical distributions enables us to achieve estimation guarantees comparable to those known for the Gaussian case.

Concretely, we propose an algorithm that, under a {very mild assumption} on the scatter matrix $\Sigma$, and given a nearly optimal number of samples $n = \tilde{O}(d^2/\varepsilon^2)$, computes in polynomial time an estimator $\hat{\Sigma}$ such that, with high probability,
\[
\left\| \Sigma^{-1/2} \hat{\Sigma} \Sigma^{-1/2} - \Id \right\|_{\text F} \le O(\varepsilon \log(1/\varepsilon))\,.
\]
This matches the best known guarantees for robust covariance estimation in the Gaussian setting.

As an application of our result, we obtain the \emph{first efficiently computable, nearly optimal robust covariance estimators} that extend beyond the Gaussian case. Specifically, for elliptical distributions satisfying the Hanson--Wright concentration inequality (including, for example, Gaussians and uniform distributions over ellipsoids), our estimator $\hat{\Sigma}$ of the covariance $\Sigma$ achieves the same error guarantee as in the Gaussian case:
\[
\left\| \Sigma^{-1/2} \hat{\Sigma} \Sigma^{-1/2} - \Id \right\|_{\text F} \le O(\varepsilon \log(1/\varepsilon))\,.
\]
Moreover, for elliptical distributions with sub-exponential tails (such as the multivariate Laplace distribution), we construct an estimator $\hat{\Sigma}$ satisfying the spectral norm bound
\[
\left\| \Sigma^{-1/2} \hat{\Sigma} \Sigma^{-1/2} - \Id \right\| \le O(\varepsilon \log(1/\varepsilon))\,.
\]
Remarkably, despite the heavier tails of such distributions, the covariance can still be estimated at the same rate as in the Gaussian case---a phenomenon unique to high dimensions and absent in low-dimensional settings.

Our approach is based on estimating the covariance of the \emph{spatial sign} (i.e., the projection onto the sphere) of elliptical distributions. The estimation proceeds in several stages, one of which involves a novel \emph{spectral covariance filtering algorithm}. This algorithm combines  covariance filtering techniques with degree-4 sum-of-squares relaxations, and we believe it may be of independent interest for future applications.

\end{abstract}

\section{Introduction}
\label{sec:introduction}

Computationally efficient robust covariance estimation is a well-known example of a problem where guarantees for the Gaussian distribution are significantly stronger than for other, even very well-behaved, distributions. In particular, the stability-based filtering algorithm for the Gaussian case from \cite{DiakonikolasKK016} relies not only on strong concentration bounds, but also on a specific algebraic relationship between the second and fourth moments of the Gaussian distribution. These assumptions are clearly quite restrictive and prevent the extension of this algorithm to any meaningful class of distributions beyond Gaussians.

Recently, significant progress has been made on \emph{non-optimal} computationally efficient robust covariance estimation for certain natural classes of distributions that generalize Gaussians. The result of \cite{sos-subgaussian} implies that dimension-independent error in relative spectral norm is achievable in polynomial time for all sub-Gaussian distributions. Similarly, \cite{sos-poincare} showed that dimension-independent error in relative Frobenius norm is achievable in polynomial time for all distributions satisfying the Poincaré inequality (with a dimension-independent constant). Their results also imply nearly optimal estimation with \emph{quasi-polynomial} number of samples (and quasi-polynomial running time).

We show that the class of \emph{elliptical distributions} is, in some sense, the right generalization of Gaussian distributions for robust covariance estimation. We hope that our work will motivate further research on elliptical distributions within algorithmic high-dimensional robust statistics.

Elliptical distributions are well-studied in the statistical literature and form a rich class of distributions.
Notable examples include the Gaussian distribution, multivariate Student's $t$-distributions (including the multivariate Cauchy), symmetric multivariate stable distributions, multivariate logistic distributions, multivariate symmetric hyperbolic distributions, and the multivariate Laplace distribution.

Scatter\footnote{Since each elliptical distribution can be written as $Ax$, where $x$ has spherically symmetric distribution, $AA^\top$ is proportional to the covariance whenever it exists. $\Sigma = AA^\top$ is called the scatter matrix. Note that it is only defined up to a multiplicative factor.} matrix estimation is not an easy problem even in the absence of corruptions. In the setting $d=\Theta(1)$, consistent estimation of the scatter matrix $\Sigma$ of an elliptical distribution is challenging, and has been studied, particularly, in \cite{Tyler1987StatisticalAF} and \cite{MAGYAR-TYLER}. 
Consistent estimation of the eigenvectors of $\Sigma$ is significantly easier, and is called \emph{elliptical component analysis} (see, for example, \cite{ECA}). 
This is because the covariance matrix of the spatial sign (the projection onto the sphere) has the same eigenvectors as $\Sigma$, and can be easily estimated, since the spatial sign has sub-Gaussian distribution. In the high-dimensional setting, the eigenvalues of $\Sigma$ can also be accurately estimated (as we show, up to factor $\Paren{1 + O(1/d)}$).


Robust scatter matrix estimation of elliptical distributions was studied\footnote{In Huber's contamination model.} in \cite{chen2018robust}. They provided a robust estimator with optimal error in spectral norm. Their estimator, however, requires exponential computation time. In this work we show that a comparable, nearly optimal error can be achieved in polynomial time.

\subsection{Problem Set-Up}\label{sec:set-up}


We use the \emph{strong contamination model} for corruptions. 
\begin{definition}\label{def:corruption}
    Let $ x_1, \ldots,  x_n$ be \iid samples from a distribution $\cD$ and let $\e > 0$.
    We say that $z_1, \ldots, z_n$ are an {$\e$-corruption} of $x_1, \ldots,  x_n$, if they agree on at least $(1-\e)n$ points.
    The remaining $\e$-fraction of the points can be arbitrary and in particular, can be corrupted by a computationally unbounded adversary with full knowledge of the model, our algorithm, and all problem parameters.
\end{definition}
This model captures other models frequently studied in the literature, in particular Huber's contamination model, where
$z_1\ldots z_n \simiid \Paren{1-\e} \cD + \e \cQ$,
where $\cQ$ is an arbitrary (unkown) distribution in $\R^d$.

Below we give a formal definition of \emph{elliptical distributions}.
\begin{definition}
    \label{def:elliptical}
    A distribution $\cD$ over $\R^d$ is {elliptical}, if for $ x \sim \cD$,
    \[
         x = \mu + \xi A  U\,,
    \]
    where $U$ is a uniformly distributed over the ($(d-1)$-dimensional) unit sphere in $\R^d$, $\xi$ is a positive random variable independent of $U$, $\mu\in \R^d$ is a fixed vector,  and $A\in \R^{d\times d}$ is a fixed linear transformation.
\end{definition}
Vector $\mu$ is called the \emph{location} of $\cD$, and the matrix $\Sigma = A\transpose{A}$ is called the \emph{scatter} matrix of $\cD$. Note that the scatter matrix is only defined up to a positive constant factor. If $\cD$ has finite covariance matrix, then $\Sigma$ is proportional to it, and if $\cD$ has finite expected value, it is equal to $\mu$.

To formulate our main result, we need a notion of the \emph{effective rank}.
\begin{definition}\label{def:effective-rank}
    Let $\Sigma$ be a positive definite matrix. Its {effective rank} is 
    \[
    \effrank(\Sigma) := \frac{\Tr{\Sigma}}{\norm{\Sigma}}\,.
    \]
\end{definition}
Note that $\effrank(\Sigma)$ is scale-invariant, and hence is the same for all scatter matrices. For $\Sigma \in \R^{d\times d}$, $1 \le \effrank(\Sigma) \le d$. In addition, $\effrank(\Sigma) \ge d / \kappa(\Sigma)$, where $\kappa(\Sigma)$ is the condition number of $\Sigma$. Furthermore, the effective rank is almost not affected by small eigenvalues of $\Sigma$: In particular, if we take the smallest half of the eigenvalues of $\Sigma$ and replace them by arbitrarily small positive values, the effective rank of the resulting matrix is at least $\effrank(\Sigma)/2$.

\paragraph{Error Metrics}
We use \emph{relative spectral norm} and \emph{relative Frobenius norm}. Relative Frobenius norm is defined as follows
\[
 \normf{\Sigma^{-1/2}\, \hat{\Sigma}\, \Sigma^{-1/2} - \Id}\,,
\]
If  it is bounded by some $\Delta$, then 
\[
\normf{\hat{\Sigma} - \Sigma}
=\normf{\Sigma^{1/2}\Paren{\Sigma^{-1/2}\, \hat{\Sigma}\, \Sigma^{-1/2} - \Id}\Sigma^{1/2}}
\le\norm{\Sigma} \cdot \Delta\,,
\]
since for all $A,B\in \R^{d\times d}$, $\Normf{AB}\le \Norm{A}\cdot \Normf{B}$. In addition, if
\[
\normf{\Sigma^{-1/2}\, \hat{\Sigma}\, \Sigma^{-1/2} - \Id} \le \Delta \le 1\,,
\]
then $ \normf{\hat{\Sigma}^{-1/2}\Sigma^{-1/2}\hat{\Sigma}^{-1/2} - \Id} \le O\Paren{\Delta}$.

Relative spectral norm is
\[
 \norm{\Sigma^{-1/2}\, \hat{\Sigma}\, \Sigma^{-1/2} - \Id}\,,
\]
Similarly, if it is bounded by some $\Delta$, then 
\[
\norm{\hat{\Sigma} - \Sigma}
\le\norm{\Sigma} \cdot \Delta\,.
\]
If
\[
\norm{\Sigma^{-1/2}\, \hat{\Sigma}\, \Sigma^{-1/2} - \Id} \le \Delta \le 1\,,
\]
then $ \norm{\hat{\Sigma}^{-1/2}\Sigma^{-1/2}\hat{\Sigma}^{-1/2} - \Id} \le O\Paren{\Delta}$.

Relative Frobenius norm is always at least as large as relative spectral norm, but it can be larger up to $\sqrt{d}$-factor.

\paragraph{Notation} We denote by $[m]$ the set $\Set{1,\ldots,m}$ and by $\log(\cdot)$ the logarithm base $e$. For matrix $A\in \R^{d\times d}$, we denote by $\Norm{A}$ its spectral norm, by $\normf{A}$ its Frobenius norm, and by $\text{vec}(A)$ the vectorization of $A$. The notations $\Omega(\cdot), O(\cdot), \gtrsim, \lesssim$ hide constant factors. $\tilde{O}(\cdot), \tilde{\Omega}(\cdot)$ hide polylogarithmic factors.

\subsection{Results}\label{sec:results}
Our main result is the following theorem.
\begin{theorem}
    \label{thm:main}
    Let $C > 0$ be a large enough absolute constant.
    Let $d,n\in \N, \e\in\R$ be such that $0 < C\log(d)/d \le \e \le 1/C$ and
    \[
    n \ge  C\cdot d^2\log^5(d)/\e^2\,.
    \]
    
    Let $\cD$ be an elliptical distribution in $\R^d$ (see \definitionref{def:elliptical}) whose scatter matrix is positive definite and satisfies 
    \[
    \effrank(\Sigma) := \frac{\Tr{\Sigma}}{\norm{\Sigma}} \ge  C\cdot \log(d)\,.
    \]
    Let $x_1,\ldots,x_n \simiid \cD$, and let $z_1,\ldots,z_n \in \R^d$ be an $\e$-corruption of $x_1,\ldots, x_n$ (see \definitionref{def:corruption}).
    
    There exists an algorithm that, given $\e$ and $z_1,\ldots,z_n$, runs in time $\poly(n)$, and  outputs  $\hat{\Sigma} \in \R^{d\times d}$ that with high probability satisfies
    \[
    \norm{\Sigma^{-1/2}\, \hat{\Sigma}\, \Sigma^{-1/2} - \Id} \le O\Paren{\e\log(1/\e)}\,,
    \]
    where $\Sigma$ is some scatter matrix of $\cD$. 

    Furthermore, if $\e \ge C\log(d)/\sqrt{d}$, then with high probability
    \[
    \normf{\Sigma^{-1/2}\, \hat{\Sigma}\, \Sigma^{-1/2} - \Id} \le O\Paren{\e\log(1/\e)}\,.
    \]
\end{theorem}

First let us discuss our assumptions. Our assumption on the effective rank $\effrank(\Sigma)$ is really mild: In particular, it is satisfied for matrices with condition $\kappa(\Sigma)\lesssim d/\log(d)$.
Our assumptions on $\e$ can be restrictive in certain settings, but they are typical in high-dimensional robust statistics. In particular, works on fast robust estimation (for example, \cite{fast, fast-covariance}) usually focus on the regime $\e = d^{-o(1)}$.

To the best of our knowledge, prior to this work, no scatter matrix estimator with dimension-independent error in relative Frobenius norm was known for general elliptical distributions. For (weaker) spectral norm error and under (weaker) Huber's contamination model, \cite{chen2018robust} showed that if $n \gtrsim d/\e$, there exists an estimator computable in \emph{exponential} time, such that  $\norm{\hat{\Sigma} - \Sigma} \le O(\norm{\Sigma}\cdot\e)$ with high probability. This error is information theoretically optimal, even for the Gaussian case. 

Due to the statistical query lower bound from \cite{SQ-bounds}, error $O\Paren{\e\log^{1-\Omega(1)}(1/\e)}$ likely cannot be achieved by polynomial time algorithms (even in spectral norm), so our error bound seems to be not improvable. While their lower bound is not for scatter estimation, but for (more complicated) covariance estimation, we show that it is possible to very accurately estimate the scaling factor (up to factor $(1 + \tilde{O}(\e)/\sqrt{d})$), and hence  scatter matrix estimation  up to error $O\Paren{\e\log^{1-\Omega(1)}(1/\e)}$ would imply the same error for robust covariance estimation in relative spectral norm, which means that our scatter matrix estimation error bound is likely not improvable for estimators computable in time $\poly(d)$.

 To the best of our knowledge, prior to this work, computationally efficient robust covariance estimation with dimension-independent error was only known for distributions with moment bounds certifiable in sum-of-squares. We show that dimension-independent errors are achievable for the \emph{scatter} matrix estimation of arbitrary elliptical distributions, without any assumptions on the moments (recall that even the first moment is not required to exist).

Further in this section we discuss some applications of \cref{thm:main}.

\subsubsection{Applications for Robust PCA}

The scatter matrix contains a lot of useful information about the distribution. This information can be used in applications, in particular, for \emph{principal component analysis}\footnote{The (non-robust) principal component analysis for elliptical distributions was first studied in \cite{ECA}.}. In particular, we can robustly learn the leading eigenvector of $\Sigma$ assuming some gap between the first and the second largest eigenvalues:

\begin{corollary}\label{cor:pca}
    Let $C > 0$ be a large enough absolute constant.
    Let $d,n\in \N, \e\in\R$ be such that $0 < C\log(d)/d \le \e \le 1/C$ and
    \[
    n \ge  C\cdot d^2\log^5(d)/\e^2\,.
    \]
    Let $\cD$ be an elliptical distribution in $\R^d$ (see \definitionref{def:elliptical}) whose scatter matrix is positive definite and satisfies 
    \[
    \effrank(\Sigma) := \frac{\Tr{\Sigma}}{\norm{\Sigma}} \ge  C\cdot \log(d)\,.
    \]
     Suppose that for some $\gamma > 0$, $\frac{\lambda_1 - \lambda_2}{\lambda_1} \ge \gamma$, where $\lambda_1$ and $\lambda_2$ are respectively the first and the second largest eigenvalues of $\Sigma$.
    
    Let $x_1,\ldots,x_n \simiid \cD$, and let $z_1,\ldots,z_n \in \R^d$ be an $\e$-corruption of $x_1,\ldots, x_n$ (see \definitionref{def:corruption}).
    
    There exists an algorithm that, given $\e$ and $z_1,\ldots,z_n$, runs in time $\poly(n)$, and  outputs  $\hat{v} \in \R^{d}$ that with high probability satisfies
    \[
    \normf{\hat{v}\hat{v}^\top - vv^\top}\le O\Paren{\frac{\e\log(1/\e)}{\gamma}}\,,
    \]
    where $v$ is the leading eigenvector of $\Sigma$. 
\end{corollary}

This result is an easy consequence of the spectral norm bound from \cref{thm:main}. 
The error bound matches the best known guarantees for the Gaussian distribution \cite{robust-pca-li, robust-pca-diak}. However, the sample complexity is worse. This is not surprising, since the algorithms for the Gaussian case were specifically designed for robust PCA. Improving the sample complexity of robust PCA for elliptical distributions is beyond the scope of this work, but we believe it to be an interesting direction for future research.

\subsubsection{Applications for Robust Covariance Estimation}
    \cref{thm:main} shows that it is possible to robustly estimate the scatter matrix. As was previously mentioned, this allows us to learn many useful things about the covariance (when it exists): the eigenvectors, the effective rank, the condition number, the ratio between the first and the second largest eigenvalues, etc. However, learning the scaling factor between the scatter matrix that we estimate and the covariance (and hence the covariance itself) might not always be possible even when there are no corruptions. For example, in the one-dimensional case, scatter matrices are just positive numbers, and the scale can be an arbitrary positive distribution, so without assumptions on higher moments, learning its covariance (when it exists) can require arbitrarily many samples. Clearly, the same issue appears in the high-dimensional case, where the scatter matrix is non-trivial and contains all scale-invariant information about the covariance. Hence we need some moment assumptions if we want to estimate the scaling factor.

    A reasonable strategy to estimate the scaling factor is to robustly estimate\footnote{Here we assume that the location of $\cD$ is zero.} $\Tr(\Sigma) = \E\norm{x}^2$ (in this subsection we denote the covariance of $x$ by $\Sigma$ and the scatter matrix that we accurately estimated by $\tilde{\Sigma}$). Since it is a one-dimensional problem, it is easy to do in polynomial time as long as $\norm{x}^2$ is well-concentrated. 
    A natural assumption for Gaussian-type concentration of $\norm{x}^2$ is the
    \emph{Hanson-Wright inequality}:
    \begin{definition}[Hanson-Wright property]\label{def:hanson-wright}
    Let $\cD$ be a distribution in $\R^d$ with mean $\mu\in \R^d$ and positive definite covariance $\Sigma\in\R^{d\times d}$, and let $C_{HW} > 0$. We say that $\cD$ satisfies \emph{$C_{HW}$-Hanson-Wright property}, if for all $A\in \R^{d\times d}$ and $t>0$, 
    \[
    \Pr_{x\sim\cD} \Paren{\Abs{z^\top Az - \Tr A}\ge t} \le 2\exp\Paren{-\frac{1}{C_{HW}} \min\Set{\frac{t^2}{\normf{A}^2}, \frac{t}{\norm{A}}}},
    \]
    where $z = \Sigma^{-1/2}(x-\mu)$.
    \end{definition} 
    In particular, isotropic distributions that satisfy this property with $C_{HW} = O(1)$ include the standard Gaussian distribution, the uniform distribution over the sphere, the uniform distribution over the Euclidean ball, and their $O(1)$-Lipschitz transformations\footnote{Of course our results can only be applied to Lipchitz transformations of these distributions that preserve the elliptical structure.}. More generally, all $O(1)$-log-Sobolev distributions also satisfy this property. Note that the Hanson-Wright property (in our formulation) is preserved under affine transformations.
    
    If $x$ satisfies the $O(1)$-Hanson-Wright property, we can estimate $\E\norm{x}^2$ up to factor $\Paren{1+\tilde{O}(\e)/\sqrt{d}}$, and it is enough for robust covariance estimation with error $\tilde{O}(\e)$ in relative Frobenius norm. 
    
    In order to estimate the covariance in relative \emph{spectral} norm, it is enough to assume that $x$ is \emph{sub-exponential}. 
    
    \begin{definition}[Sub-exponential distributions]\label{def:sub-gaussain}
 Let $\cD$ be a distribution in $\R^d$ with mean $\mu\in\R^d$, and let $C_{SE} > 0$. We say that $\cD$ is \emph{$C_{SE}$-sub-exponential}, if for all $u\in \R^{d}$ and even $p\in \N$, 
    \[
    \Paren{\E_{x\sim \cD}\iprod{u, x-\mu}^p}^{1/p}\le C_{SE}\cdot {p}\cdot \Paren{\E_{x\sim \cD}\iprod{u, x-\mu}^2}^{1/2}\,.
    \]
    \end{definition}
    $O(1)$-sub-exponential distributions satisfy the $O(1)$-Hanson-Wright property, but the converse is not true.
    For distributions that are {simultaneously} sub-exponential and elliptical (in particular, the multivariate Laplace distribution), we can achieve nearly optimal error in relative spectral norm in polynomial time.
    Specifically, we prove the following theorem.

    \begin{theorem}\label{thm:covariance-estimation}
    Let $C > 0$ be a large enough absolute constant.
    Let $d,n\in \N, \e\in\R$ be such that $0 < C\log(d)/d \le \e \le 1/C$ and
    \[
    n \ge  C\cdot d^2\log^5(d)/\e^2\,.
    \]
    
    Let $\cD$ be an elliptical distribution in $\R^d$ (see \definitionref{def:elliptical}) with mean $\mu$ and positive definite covariance $\Sigma$ that satisfies 
    \[
    \effrank(\Sigma) := \frac{\Tr{\Sigma}}{\norm{\Sigma}} \ge  C\cdot \log(d)\,.
    \]
    Suppose in addition that $\cD$ is also $O(1)$-sub-exponential (see \cref{def:sub-gaussain}).
    Let $x_1,\ldots,x_n \simiid \cD$, and let $z_1,\ldots,z_n \in \R^d$ be an $\e$-corruption of $x_1,\ldots, x_n$ (see \definitionref{def:corruption}).
    
    There exists an algorithm that, given $\e$ and $z_1,\ldots,z_n$, runs in time $\poly(n)$, and outputs $\hat{\Sigma} \in \R^{d\times d}$ that with high probability satisfies
    \[
    \norm{\Sigma^{-1/2}\, \hat{\Sigma}\, \Sigma^{-1/2} - \Id} \le O\Paren{\e\log(1/\e)} \,.
    \]
    
    Furthermore, if $\e \ge C\log(d)/\sqrt{d}$ and $\cD$ satisfies the $O(1)$-Hanson-Wright property, then with high probability
    \[
    \normf{\Sigma^{-1/2}\, \hat{\Sigma}\, \Sigma^{-1/2} - \Id} \le O\Paren{\e\log(1/\e)}\,.
    \]

    \end{theorem}
    While it might seem surprising that we have $O(\e\log(1/\e))$ error for sub-exponential distributions instead of $O(\e\log^2(1/\e))$ that we expect to have in the one-dimensional case for (possibly symmetric) sub-exponential distributions,
    this happens due to the high-dimensional nature of the problem. Recall that $x\sim \cD$ with location $0$ can be written as $ \xi A \sqrt{d} \cdot U$, where $U$ is uniformly distributed over the unit sphere, and assume that the singular values of $A$ are $\Theta(1)$. 
    In high-dimensions, $\E\iprod{\sqrt{d} U, v}^p$ is very close to $\E_{g\sim\cN(0,\Id_d)}\iprod{g, v}^m$ unless $p$ is very large (of order $d$). Therefore, the variable $A\sqrt{d} \cdot U$ has sub-Gaussian moments\footnote{We call a distribution $\cD$ $O(1)$-\emph{sub-Gaussian} if for all $u\in \R^d$ and all even $p\in\N$, $\Paren{\E_{x\in \cD}\iprod{u, x-\mu}^p}^{1/p}\le O(\sqrt{p})\Paren{\E_{x\in \cD}\iprod{u, x-\mu}^2}^{1/2}$. Such distributions are also sometimes called \emph{hypercontractive sub-Gaussian}.}, and hence the moments of $\xi$ have to be sub-Gaussian (at least for $p \le d$), since $x$ is sub-exponential. On the other hand, since $\norm{A \sqrt{d} U}$ is a bounded random variable, $\norm{x} = \xi  \norm{A \sqrt{d} U}$ has \emph{sub-Gaussian moments}, so it is possible to robustly learn $\E \norm{x}^2$ up to factor $\Paren{1 + O(\e\log(1/\e))}$. 
    
Let us compare our results with prior work. As was previously mentioned, \cite{DiakonikolasKK016} developed a technique that achieves error $O(\e\log(1/\e))$ in relative Frobenius norm in $\poly(n)$ time with $n=\poly(d)$ samples for Gaussians.
This result exploits exact algebraic relations between Gaussian second and fourth moments, and hence is only valid for the Gaussian distribution\footnote{It can also be extended for certain mixtures of Gaussians, see section 4.3.2 of \cite{DK_book}.}.

Recently \cite{sos-poincare} showed that moment bounds on quadratic forms of arbitrary Poincar\'e distributions can be certified by (low degree) sum-of-squares proofs. \cite{Bakshi} showed that it implies that the error $O\Paren{t^2\e^{1-1/t}}$ in relative Frobenius norm is achievable for robust covariance estimation of all Poincare distributions with $d^{\poly(t)}$ samples in time $d^{\poly(t)}$. In particular, there exists an estimator that achieves error $\tilde{O}(\e)$ for all Poincar\'e distributions with $d^{\polylog(1/\e)}$ time and sample complexity.
Another recent result \cite{sos-subgaussian} shows that moment bounds of arbitrary sub-Gaussian distributions are certifiable in sum-of-squares. \cite{KS17} showed that it implies the guarantees similar to those of \cite{Bakshi} for all sub-Gaussian distributions, but in (weaker) relative spectral norm. Note that both results require \emph{quasi-polynomial} number of samples (and running time) for error $\tilde{O}(\e)$. 

Elliptical distributions admit much better estimation guarantees than sub-Gaussian or Poincaré distributions: Our estimator achieves nearly optimal error $\tilde{O}(\varepsilon)$, with a sample complexity that is nearly optimal, and runs in polynomial time.

\section{Techniques}
\label{sec:techniques}

Let us recall the assumptions that are used for nearly optimal computationally efficient robust covariance estimation in the Gaussian case (see section 5.2 of \cite{DiakonikolasKK016} or chapter 4 of \cite{DK_book} for the description of the algorithm). First, we need very strong concentration assumptions. In particular, $O(1)$-sub-Gaussianity is not enough, while $O(1)$-Hanson-Wright property is enough for the analysis to work (with $\tilde{O}(d^2/\e^2)$ samples). 
Second, we need the fourth moment of $\Sigma^{-1/2}x$ to coincide with the fourth moment of the standard Gaussian distribution. 

For elliptical distributions we do not have any assumptions on the moments (covariance or even mean might not even exist), and we have to rely only the elliptical structure. Fortunately, this structure allows us to reduce the problem of scatter matrix estimation to a problem of robust covariance estimation of some specific $O(1)$-sub-Gaussian distribution. The properties of this distribution depend on the spectrum of $\Sigma$, and, as we show, the closer $\Sigma$ is to $\Id$, the better (for the robust covariance estimation) these properties are. This allows us to estimate the covariance in several steps that we describe below. 

\paragraph{Spatial Sign.} The distribution reverenced above is the \emph{spatial sign} of an elliptical distribution. It is a projection of an elliptical vector with location zero\footnote{We can always reduce the problem to this case by considering $\paren{x_{i}-x_{\lfloor n/2\rfloor+i}}/\sqrt{2}$. The resulting samples also come from (an $O(\e)$-corruption of) an elliptical distribution with the same scatter matrix.} onto some (arbitrary) sphere centered at $0$. In this paper we use the sphere of radius $\sqrt{d}$, and denote the projection of vector $x\in\R^d$ onto this sphere by $\spsign(x)$. 

If $x$ is an elliptical vector with location $0$ and scatter matrix $\Sigma$, $\spsign(x)$ depends \emph{only} on $\Sigma$, and is the same for \emph{all} elliptical distributions with the same scatter matrix. Indeed, since  $x = \xi A U$ (as in Definition \ref{def:elliptical}), the projection onto the sphere satisfies
\[
 \spsign(x) = \frac{x}{\tfrac{1}{\sqrt{d}}\norm{x}} = \frac{\xi A U}{\tfrac{1}{\sqrt{d}}\norm{\xi A U}} = \frac{A U}{\tfrac{1}{\sqrt{d}}\norm{A U}}\,.
\]
So when we study the properties of $\spsign(x)$, we can simply assume that $x$ comes from our favorite elliptical distribution: $\cN(0,\Sigma)$.

The spatial sign was extensively studied in prior works on elliptical distributions, because $\Sigma' := \Cov_{x\sim \cN(0,\Sigma)} \spsign(x)$ has the same eigenvectors as $\Sigma$, and hence $\Sigma'$ is very useful for the principal (elliptical) component analysis.
However, the eigenvalues of $\Sigma'$ differ from the eigenvalues of $\Sigma$, so even in the classical setting (without corruptions), if we use the empirical covariance of the spatial sign to estimate the eigenvalues of $\Sigma$ (or $\Sigma$ itself), the error of the estimator is not vanishing, even if we take infinitely many samples. 
Fortunately, in robust statistics we anyway have a term that does not depend on the number of samples (it should be at least $\Omega(\e)$ in the Gaussian case). So if we prove that $\Sigma'$ is $\tilde{O}(\e)$-close to $\Sigma$ (in relative Frobenius or at least relative spectral norm), then good robust estimators of $\Sigma'$ are also good robust estimators of $\Sigma$. Since $\Tr(\Sigma') = d$, we fix the scale of $\Sigma$ so that $\Tr(\Sigma) = d$.

Let us discuss how to bound $\norm{{\Sigma}^{-1/2}\, \Sigma'\, {\Sigma}^{-1/2} - \Id}$ (then we also obtain the bound on relative Frobenius norm by multiplying the spectral norm bound by $\sqrt{d}$). Since the eigenvectors of $\Sigma$ and $\Sigma'$ are the same, we can work in the basis where both matrices are diagonal. It follows that
\[
\norm{{\Sigma}^{-1/2}\, \Sigma'\, {\Sigma}^{-1/2} - \Id} = \max_{i\in[d]}\Abs{\frac{\lambda'_i}{\lambda_{i}} - 1} = \max_{i\in[d]}\Abs{\E_{x\sim \cN(0,\Sigma)}{\frac{x_{j}^2/\lambda_j}{\tfrac{1}{d}\norm{x}^2} - 1}}
\]
where $\lambda_i = \Sigma_{ii}$, and $\lambda'_i = \Sigma'_{ii}$.
The Hanson-Wright inequality implies that if $\effrank(\Sigma) \gtrsim \log(d)$, then with high probability for all of the $n = \poly(d)$ samples $x_1,\ldots, x_n\simiid \cN(0,\Sigma)$, 
\[
\norm{x_i}^2 = d\cdot \Paren{1 \pm O\Paren{\sqrt{\frac{\log d}{\effrank(\Sigma)}}}}\,.
\]

Assuming that it holds with probability $1$\footnote{In the formal argument, we analyze not exactly the spatial sign, but some function that with high probability coincides with it on all of the samples, and for this function this statement is true.}, we get a bound
$\norm{{\Sigma}^{-1/2}\, \Sigma'\, {\Sigma}^{-1/2} - \Id} \le O\Paren{\sqrt{\frac{\log d}{\effrank(\Sigma)}}}$. While this bound can be useful for (non-robust) eigenvector estimation\footnote{It was used, in particular, in \cite{ECA}, and it is enough for consistent estimation of the eigenvectors of $\Sigma$.}, for our purposes it is too bad: Even if the effective rank is as large as possible ($\effrank(\Sigma) \ge \Omega(d)$), we only get error $O(\log(d))$ in relative Frobenius norm.

To get a better bound, we get rid of the denominator by expanding it as a series: 
\[
\E\frac{x_{j}^2/\lambda_j}{\tfrac{1}{d}\norm{x}^2} = 
\sum_{k=0}^{\infty} \E{\frac{x_j^2}{\lambda_j} \Paren{1 - \tfrac{1}{d}\norm{x}^2}^k}
=\sum_{k=0}^{\infty} \E{g_j^2 \Paren{1 - \tfrac{1}{d}\sum_{i=1}^{d}\lambda_i g_i^2}^k}
\,,
\]
where $g = \Sigma^{-1/2}x \sim \cN(0,\Id)$. The term that corresponds to $k=0$ is $1$. The term that corresponds to $k=1$ is
\[
\E g_j^2\paren{1 - \tfrac{1}{d}\sum_{i=1}^{d}\lambda_i g_i^2} = 
1 - \tfrac{1}{d}\sum_{i\neq j} \lambda_i - \frac{3\lambda_j}{d} =
1 - \tfrac{1}{d}\sum_{i=1}^d \lambda_i - \frac{2\lambda_j}{d}
= \frac{2\lambda_j}{d} \le O\Paren{\frac{1}{\effrank(\Sigma)}}\,,
\]
where we used $\sum_{i=1}^d \lambda_i = \Tr(\Sigma) = d$. Similarly, the term that corresponds to $k = 2$ is also $O\Paren{\frac{1}{\effrank(\Sigma)}}$, and the other terms are much smaller. Hence we get a bound
\[
\norm{{\Sigma}^{-1/2}\, \Sigma'\, {\Sigma}^{-1/2} - \Id} \le O\Paren{\frac{1}{\effrank(\Sigma)}}\,.
\]
In the case $\effrank(\Sigma)\ge\Omega(d)$, we get $O(1/d)$ error in relative spectral norm, and $O(1/\sqrt{d})$ error in relative Frobenius norm. This is the reason why we require the lower bounds on $\e$ in \cref{thm:main}: We want to make these errors smaller than the robust estimation error $\tilde{O}(\e)$. However, if the effective rank is small, this error is still too large: For example, the error bound in relative Frobenius norm is $\Omega(1)$ if $\effrank(\Sigma) \le \sqrt{d}$. 

Furthermore, if $\effrank(\Sigma) \le o(d)$, the $\paren{d^2\times d^2}$-dimensional covariance of the $d^2$-dimensional random variable $\Paren{\Sigma'}^{-1/2}\spsign(x)\spsign(x)^\top\Paren{\Sigma'}^{-1/2}$ might not even be bounded by $O(1)$ (in spectral norm), which makes robust estimation of $\Sigma'$ with dimension-independent error in  relative Frobenius norm challenging, even if we did not aim to achieve nearly optimal error.

In order to fix these issues, first we estimate $\Sigma'$ up to some small constant error in relative spectral norm. For this, we split the sample into several (more precisely, 3) sub-samples, and for each new estimation we use fresh samples.

\paragraph{First Estimation.} It is not difficult to see that as long as $\effrank(\Sigma)\gtrsim \log(d)$, then $\spsign(x)$ is $O(1)$-sub-Gaussian. By recent result \cite{sos-subgaussian}, the bound on its fourth moment can be certified via a degree-$4$ sum-of-squares proof\footnote{This fact can be also easily verified directly without referring to \cite{sos-subgaussian}.}. Hence we can use the sum-of-squares algorithm from \cite{KS17} that estimates $\Sigma'$ in relative spectral norm up to error $O(\sqrt{\e})$ (this algorithm works if we use $n\gtrsim d^2\log^2(d)/\e^2$ samples). Since $\Sigma'$ is close to $\Sigma$, with high probability we get an estimator $\hat{\Sigma}_1$ such that 
\[
\norm{\hat{\Sigma}_1^{-1/2} \Sigma \hat{\Sigma}_1^{-1/2} - \Id}  \le O\paren{\sqrt{{\e}} + \frac{1}{\effrank(\Sigma)}}\,.
\]
In particular, $0.9 \cdot \Id \preceq \hat{\Sigma}_1^{-1}\Sigma \preceq 1.1\cdot \Id$. 
Hence if we multiply the next subsample by $\hat{\Sigma}_1^{-1}$, we get ($O(\e)$-corruption of) samples from an elliptical distribution with new scatter matrix $\tilde{\Sigma} = \rho \hat{\Sigma}_1^{-1} \Sigma$, where $\rho = d/\Tr(\hat{\Sigma}_1^{-1}\Sigma)$. 

So we can assume that we work with (corrupted) samples from an elliptical  distribution $\cD$ with scatter matrix $\Sigma$ such that $0.9 \cdot \Id \preceq \Sigma \preceq 1.1\cdot \Id$. If we fix the scale $\Tr(\Sigma) = d$, then
$\Sigma' = \Cov_{x\sim \cD}\spsign(x)$ is $O(1/d)$-close to $\Sigma$ in relative spectral norm (since $\effrank(\tilde{\Sigma}) \ge \Omega(d)$).

Now we can try to estimate $\Sigma'$ up to error $\tilde{O}(\e)$. As was previously mentioned, the algorithm for the Gaussian distribution from \cite{DiakonikolasKK016} requires strong assumptions: Hanson-Wright concentration, and the Gaussian fourth moment. 
Fortunately, as long as $\effrank(\Sigma)\ge \Omega(d)$, $\spsign(x)$ satisfies the $O(1)$-Hanson-Wright property. Indeed, it is not hard to see that with overwhelming probability $\Paren{\Sigma'}^{-1/2}\spsign(x) = \Paren{\Sigma'}^{-1/2}\spsign(\Sigma^{1/2} g)$ (where $g\sim \cN(0,\Id)$) coincides with some $O(1)$-Lipschitz function of $g$ (concretely, a composition of linear transformations with with bounded spectral norm, and a function that projects onto the sphere only the points that are close to it, and is linear otherwise). It is known that\footnote{See, for example, \cite{log-sobolev-are-hanson-wright}.}  any $O(1)$-Lipschitz function of a standard Gaussian vector satisfies the $O(1)$-Hanson-Wright property. Therefore, we do not have to worry about the concentration, and can focus on dealing with the fourth moment.

Recall that the covariance filtering algorithm with nearly optimal error uses the $(d^2\times d^2)$-dimensional covariance of the distribution in the isotropic position: $T=\E\Paren{\Cov(y)^{-1/2} yy^\top\Cov(y)^{-1/2} - \Id_d}^{\otimes 2}$. If $y\sim \cN(0,\Sigma)$, then $Q = 2\Id_{d^2}$. However, in our case $y = \spsign(x)$, and the situation is more complicated. This covariance is not only different from $2\Id_{d^2}$, it is also unknown to us, since it depends on $\Sigma$.
We study this dependence and show that the entries $T_{ijij}$ and $T_{iijj}$ are $O\Paren{\Norm{\Sigma - \Id_d}/d + \tilde{O}(1/d^2)}$-close to the entries of  $S := \E\Paren{gg^\top / \norm{g}^2 -\Id_d}^{\otimes 2}$ (and the other entries are zero for both of them). While at the first glance $T$ and $S$ seem to be very close, it is in fact not true: their values (as quadratic forms on unit vectors in $\R^{d^2}$) can differ by $O(\Norm{\Sigma - \Id_d})$. Even if we again use the algorithm from \cite{KS17} and guarantee that $\norm{\Sigma - \Id_d}\le O(\sqrt{\e})$, this bound is only $O(\sqrt{\e})$, while we need it to be $\tilde{O}(\e)$. 
Hence we have to somehow estimate $\Sigma'$ up to error $\tilde{O}(\e)$ in \emph{spectral} norm, before estimating it in  Frobenius norm. 

\paragraph{Second Estimation.}
Let us first describe why the difference between the values of $T$ and $S$ on unit $V\in\R^{d^2}$ can be $O(\Norm{\Sigma - \Id_d})$. The reason is the terms $V_{ii}V_{jj} \Paren{T_{iijj} - S_{iijj}}$ for $i\neq j$. Since $\sum_{i=1}^d V_{ii}$ can be as large as $\sqrt{d}$, the sum of $V_{ii}V_{jj} \Paren{T_{iijj} - S_{iijj}}$ can be as large as $d \cdot \normi{S-T} = O(\Norm{\Sigma - \Id_d})$. 

However, if we only consider unit $d^2$-dimensional vectors of the form $V = uu^\top$, this issue disappears. Indeed, $\sum_{i=1}^d V_{ii}$ is now bounded by $1$, and $T$ is $O(1/d)$-close to $S$ on such vectors. Furthermore, they are both $O(1/d)$-close (as quadratic forms on unit $d^2$-dimensional vectors of the form $uu^\top$) to $2\Id_{d^2}$.

In fact, if we use the covariance filtering algorithm only for such vectors, we get the desired spectral norm bound. However, the optimization problem involved is computationally hard, since we need to optimize over the set $\Set{u^{\otimes 4} \suchthat \norm{u} = 1}$.
Fortunately, we can work with the (canonical) \emph{sum-of-squares relaxation} of this set, that is, with the set of degree-$4$ pseudo-expectations of $v^{\otimes 4}$ that satisfy the constraint $\norm{u}^2 =  1$. 

In order to show that this approach works, we introduce a generalized notion of stability that we use not only for estimation in spectral norm, but also in Frobenius form at the final step. 

\begin{definition}[Generalized Stability]\label{def:stability}
    Let $m\in \N$, $\cV \subseteq \R^m$, $\cP \subseteq \R^{m\times m}$ such that $\cV \otimes \cV \subseteq \cP$, $\mu\in \R^m$ and $Q\in \R^{m\times m}$. Let $\e, \delta, r > 0$ such that $\delta \ge \e$. 
    
    A finite multiset $M$ of points from $\R^m$ is \emph{$\Paren{\e,\delta,r,\cV,\cP}$-stable with respect to $\mu$ and $Q$} if for every $v\in \cV$, every $P\in \cP$, and every $M'\subseteq M$ with $\Card{M'}\ge \paren{1-\e}\Card{M}$, the following three conditions hold:
    \begin{enumerate}
        \item $\Abs{\tfrac{1}{\card{M'}} \sum_{x\in M'} \iprod{v, x-\mu}}\le \delta$,
        \item $\Abs{\tfrac{1}{\card{M'}} \sum_{x\in M'} \iprod{P, \paren{x-\mu}\paren{x-\mu}^\top} - \iprod{P, Q}}\le \delta^2/\e$, and
        \item $\Abs{\iprod{P, Q}} \le r^2$.
    \end{enumerate}
\end{definition}

Note that we use it with $m = d^2$ and apply it to the set $S$ of samples $y_1y_1^\top,\ldots, y_ny_n^\top$, where $y = \spsign(x)$.
For estimation in spectral norm, $\cV$ is $\Set{uu^\top \suchthat \norm{v} = 1}$, $\cP$ is the set of $\pE u^{\otimes 4}$ described above, $Q = 2\Id_{d^2}$, and $r = \sqrt{2}$ (and, as in the standard filtering algorithm, $\delta = O(\e\log(1/\e)$). 

As previously mentioned, if an isotropic distribution $\cY$ satisfies the Hanson-Wright property, then it can be shown that the set of iid samples drawn from this distribution is $(\e,O(\e\log(1/\e)), O(1), \cB, \Conv(\cB\otimes \cB))$-stable with respect tp $\mu = \Id_d$ and $Q = \E_{y\sim \cY}\Paren{yy^\top - \Id_d}^{\otimes 2}$ with high probability, where $\cB = \{V\in \R^{d^2} : \normf{V} = 1\}$. Since our $\cV$ is a subset of $\cB$, and our $\cP$ is a subset of $\Conv(\cB\otimes \cB)$, we also get the stability for $\cV$ and $\cP$.

The filtering algorithm works in a similar way to the Gaussian case: It assigns weights to the samples, in the case of the covariance estimation transforms them\footnote{For the covariance filtering the samples are transformed via linear transformation $y_iy_i^\top \mapsto C^{-1/2}y_iy_i^\top C^{-1/2}$, where $C$ is the current candidate for the covariance estimator.}, and checks whether the value $\lambda := \max_{P\in\cP}\Abs{\iprod{P, Q-\hat{Q}^{(t)}}}$  is too large
(where $\hat{Q}^{(t)}$ is the weighted empirical $m\times m$-dimensional covariance of the (transformed) samples). 
In particular, it requires $Q$ (but, of course, not $\mu$) to be known.
If $\lambda$ is large, it reassigns the weights, so that the weights of the samples that made it large decrease. In the end $\lambda$ should be small, and if it is small, then the stability guarantees that the current weighted sample mean is close to $\mu$.
Note that since we can optimize linear functions over $\cP$ efficiently, the algorithm runs in polynomial time.

This notion of stability, however, is not enough for the filtering algorithm to work. Apart from the condition that  the set of samples is stable at each iteration\footnote{The transformed set needs to be stable for certain $\delta' > \delta$. This condition also appears in the Gaussian setting, and for our case can be shown in a similar way.} of the algorithm, we need some additional conditions that $\cP$ is in certain sense not much larger than $\cV\otimes \cV$. Concretely, for all $A,B\in \R^{d^2}$ and $P\in \cP$, we need to show that
\[
\Iprod{A\otimes B, P}^2 \le \Paren{\sup_{V\in \cV}\Iprod{A, V}^2}\cdot\Paren{\sup_{V\in \cV}\Iprod{B, V}^2}\,.
\]
and $\Iprod{A\otimes A, P} \ge 0$.
For $P = v^{\otimes 4}$, it is possible to show that these inequalities can be certified via degree-$4$ sum-of-squares proofs, so $\cP$ satisfies the desired properties. 

Thus, we show that the filtering algorithm for spectral covariance estimation finds $\hat{\Sigma}_2$ such that 
$\norm{\hat{\Sigma}_2^{-1/2} \Sigma \hat{\Sigma}_2^{-1/2} - \Id} \le O\Paren{\e\log(1/\e)}$. After we multiply fresh samples by $\hat{\Sigma}_2^{-1/2}$, we can assume that we work with a sample from an elliptical distribution with scatter matrix $\Sigma$ such that $\Norm{\Sigma - \Id} \le O(\e\log(1/\e))$.

\paragraph{Final Estimation.}
To estimate $\Sigma$ in Frobenius norm, we again use the stability-based filtering. We use it with $Q = S$ (recall that $S := \E\Paren{gg^\top / \norm{g}^2 -\Id_d}^{\otimes 2}$).
Since $\Sigma$ is very close to $\Id_d$, the situation is simpler than before: we show the stability of the \emph{non-transformed} samples (with the same $\delta = O(\e\log(1/\e)$), and then simply apply filtering without transformations (in other words, this filtering algorithm is oblivious to the matrix structure in $\R^{d^2}$).  This algorithm finds $\hat{\Sigma}_3$ such that $\normf{\hat{\Sigma}_3^{-1/2} \Sigma \hat{\Sigma}_3^{-1/2} - \Id} \le O\Paren{\e\log(1/\e)}$. 

Now let us use the notation $\Sigma$ for the original scatter matrix with $\Tr(\Sigma) = d$, 
$\Sigma_1 = \rho_1\hat{\Sigma}_1^{-1/2}\Sigma\hat{\Sigma}_1^{-1/2}$, 
$\Sigma_2 = \rho_2\hat{\Sigma}_2^{-1/2}\Sigma_1\hat{\Sigma}_2^{-1/2}$, where $\rho_1$ and $\rho_2$ are chosen so that $\Tr(\Sigma_1) = \Tr(\Sigma_2) = d$. Then $\hat{\Sigma} := \hat{\Sigma}_1 \hat{\Sigma}_2 \hat{\Sigma}_3$
satisfies
\[
\normf{\hat{\Sigma}^{-1/2} \Paren{\rho_1\rho_2\Sigma} \hat{\Sigma}^{-1/2}  - \Id}\le O(\e\log(1/\e))\,.
\]
So $\hat{\Sigma}$ is the desired estimator of the scatter matrix $\rho_1\rho_2\Sigma$.


\section{Related Work and Future Directions}

\subsection{Related Work}
We refer to \cite{DK_book} for a survey on the area of high-dimensional robust statistics. The following is a non-exhaustive selection of related work.

The filtering algorithm was proposed in \cite{DiakonikolasKK016}, and was especially useful in the Gaussian setting for various robust estimation problems. Sum-of-squares was first used in robust statistics in \cite{KS17} and \cite{HL18}. Apart from robust covariance estimation, sum-of-squares approach in robust statistics was used for robust mean estimation \cite{KS17, pmlr-v167-kothari22a, robust-sparse-mean-estimation, neurips-symmetric}, robust regression \cite{KlivansKM18, robust-regression-bakshi, liu2024robust}, robust sparse and tensor PCA \cite{d2023higher}, robust clustering \cite{HL18, BDH+20, BDH+22, dmitriev2024robust}.


\subsection{Future Directions}

As was mentioned before, one promising direction is optimizing the sample complexity of robust PCA for elliptical distributions. 

Another promising direction is to determine whether it is possible to estimate the \emph{location} parameter with $\poly(d)$ samples in time $\poly(d)$ up to error $\tilde{O}(\e)$. For spherically symmetric distributions it was done in \cite{neurips-symmetric}. However, it is not clear how to use our result in their analysis. There are several things that do not work with their analysis if the distribution is not exactly spherically symmetric. As an example, as a part of their analysis, they show a lower bound for all small enough vectors $\Delta$:
\[
\sum_{i\in J} f(x_i + \Delta)f(x_i + \Delta)^\top \succeq \E f(x)f(x) - \tilde{O}(\e)\Id_d\,,
\]
where $f_j: x_{ij}\mapsto f_j(x_{ij})$ are bounded  Lipschitz functions (applied entry-wise), and $J$ is a subset of $[n]$ of size at least $(1-\e)n$.
For spherically symmetric $x$, it is possible to condition on the absolute values, and use the fact that the signs are independent of absolute values and are iid. It implies that that the conditional distribution is sub-Gaussian, and it is enough to derive this bound. However, if $x$ is elliptical and not exactly spherically symmetric, this argument does not work. 

Although it is not clear to us how to extend the analysis to the non-spherical setting, we conjecture that nearly optimal and efficient\footnote{\cite{chen2018robust} showed that \emph{Tukey's median} estimator (computable in exponential time) $O(\e)$-close to the location (in Huber's contamination model).} estimation of the location parameter should be possible.

\phantomsection
\addcontentsline{toc}{section}{References}
\bibliographystyle{amsalpha}
\bibliography{bib/custom,bib/dblp,bib/mathreview,bib/scholar}

\newpage
\appendix

\section{First Estimation}
First, if we have more samples than needed, we simply take only the necessary number of samples (choosing them at random) and do not use the other samples. We need this step to make sure that $n=\poly(d)$, and small probability events do not happen with any of the (uncorrupted) samples.
For simplicity, further we denote by $n$ the number of samples that the algorithm actually uses.

For $i \in [\lfloor n / 2 \rfloor]$, let $z'_i = \paren{z_i - z_{\lfloor n / 2 \rfloor + i}}/\sqrt{2}$. Then, by Theorem 4.1 from \cite{frahm2004generalized}, $z'$ is a $4\e$-corruption of the sample $x'_1,\ldots x'_{\lfloor n/2 \rfloor} \simiid \cD'$, where $\cD'$ is an elliptical distribution with location $0$ and scatter matrix $\Sigma$. To simplify the notation, we assume that this preprocessing step has been done, and the input $z_1,\ldots,z_n$ is an $\e$-corruption of an iid sample from an elliptical distribution with location $0$ and scatter matrix $\Sigma$.

Then we project $z_1,\ldots, z_n$ onto the sphere of radius $\sqrt{d}$. 
We can assume that the projected samples are $\spsign(y'_1),\ldots\spsign(y'_n)$, where $y'_1,\ldots y'_n$ is an $\e$-corruption of $y_1,\ldots, y_n\simiid \cN(0,\Sigma)$ (and $\spsign$ is the projection). 
Even though $\spsign(y)$ has nice properties, for technical reasons it is  not convenient to analyze it. Instead, we analyze the function $\cF:\R^d \to \R^d$ defined as
    \[
    \cF(x) = 
    \begin{cases}
        \spsign(x) = {\sqrt{d}} \cdot x/{\norm{x}} &\text{if } 0.9d \le \norm{x}^2 \le 1.1d\,,\\
        \sqrt{10/9}\cdot x &\text{if } \norm{x}^2 < 0.9d\,,\\
        \sqrt{10/11} \cdot x &\text{if } \norm{x}^2 > 1.1d\,.
    \end{cases}
    \]

By Hanson-Wright inequalities (\cref{prop:Hanson-Wright}), if $\effrank(\Sigma) \gtrsim \log(d)$, then with high probability for all $i\in[n]$, $\cF(y_i) = \spsign(y_i)$. 
Hence we can assume that we are given an $\e$-corruption of iid samples that are distributied as $\cF(y)$, where $y\sim \cN(0,\Sigma)$. We call these corrupted samples $\zeta_1,\ldots, \zeta_n$, and true samples $\cF(y_1),\ldots , \cF(y_n)$.

The following lemma shows that the covariance of $\cF(y)$ is close to $\Sigma$.

\begin{lemma}\label{lem:closeness}
    Let $\Sigma \in \R^{d\times d}$ be a positive semidefinite matrix such that $\Tr(\Sigma) = d$ and $\effrank(\Sigma) \gtrsim \log(d)$.
    
    Then 
    \[
    \Norm{\Sigma^{-1/2}\,
    \Paren{\Cov_{y \sim \cN(0,\Sigma)}{\cF(y)} - \Sigma}
    \,\Sigma^{-1/2}}  \le O\Paren{\frac{1}{\effrank(\Sigma)}}\,.
    \]

\end{lemma}
\begin{proof}
    Observe that $\Sigma' := \Cov_{y \sim \cN(0,\Sigma)}\brac{\cF(y)}$ has the same eigenvectors as $\Sigma$. Indeed, in the basis where $\Sigma$ is diagonal, 
    \[
    \Sigma'_{ij}= \E\Brac{\frac{\sqrt{\lambda_i\lambda_j} \cdot g_ig_j}{\tfrac{1}{d}\sum_{i=1}^d \lambda_i g_i^2 \cdot \ind{\sum_{i=1}^d \lambda_i g_i^2\in[0.9d,1.1 d]} 
    + \tfrac{9}{10} \ind{\sum_{i=1}^d \lambda_i g_i^2 < 0.9d} 
    + \tfrac{11}{10} \ind{\sum_{i=1}^d \lambda_i g_i^2 > 1.1d}
    }} 
    \,,
    \]
    where $\lambda_1,\ldots, \lambda_d$ are the eigenvalues of $\Sigma$. Since the signs of $g_i$ and $g_j$ are independent and are independent of all $g_1^2,\ldots, g_d^2$, $\Sigma'_{ij} = 0$ if $i\neq j$.
    
    Hence if we take these eigenvectors as the basis, the matrices $\Sigma^{-1/2}$, $\Sigma'$ and $\Sigma$ become diagonal, and the value of the relative spectral norm does not change. In this basis, the eigenvalues of $\Sigma'$ are its diagonal entries.
    
    Note that
       \[
     \Norm{\Sigma^{-1/2}\,
    \Paren{\Sigma' - \Sigma}
    \,\Sigma^{-1/2}} = \max_{j\in [d]} \Abs{\frac{\lambda_j'}{\lambda_j} - 1} \,.
    \]

    Consider 
    \[
    q = 1 - \Paren{\tfrac{1}{d}\sum_{i=1}^d \lambda_i g_i^2 \cdot \ind{\sum_{i=1}^d \lambda_i g_i^2\in[0.9d,1.1 d]} 
    + \tfrac{9}{10} \ind{\sum_{i=1}^d \lambda_i g_i^2 < 0.9d} 
    + \tfrac{11}{10} \ind{\sum_{i=1}^d \lambda_i g_i^2 > 1.1d}}\,.
    \]
    Note that since $\abs{q} \le 1/10$, $\frac{1}{1 - q} = \sum_{k=0}^{\infty} q^k$.
    Therefore,
    \[
    \frac{\lambda'_i}{\lambda_i}
    =
    \sum_{k=0}^{\infty} \E{g_j^2 q^k} 
    = \E{g_j^2}+ \E\Brac{g_j^2 q} +  \E \Brac{g_j^2 q^2}  + \sum_{k=3}^{\infty} \E \Brac{g_j^2 q^k}\,.
    \]
    The first term is $1$. 
    Let us bound the second term. Let $N=d^{1000}$, and for all $i\in[d]$ consider $g_{1i},\ldots,g_{Ni}\simiid \cN(0,1)$. By Hanson-Wright inequalities (Proposition \ref{prop:Hanson-Wright}) with high probability, for all $m \in [N]$,  
    $\tfrac{1}{d}\sum_{i=1}^d \lambda_i (1-g_{mi}^2) \in [0.9,1.1]$. 
    Also, with high probability, the sample average of 
    $g_{mj}^2 \cdot \tfrac{1}{d}\sum_{i=1}^d \lambda_i (1-g_{mi}^2)$
    is $d^{-10}$-close to both $\E\Brac{g_j^2 q}$ and $\E g_j^2 \cdot\tfrac{1}{d}\sum_{i=1}^d \lambda_i (1-g_i^2) = \frac{2\lambda_j}{d}$ (since both random variables $g_j^2 q$ and $g_j^2 \cdot\tfrac{1}{d}\sum_{i=1}^d \lambda_i (1-g_i^2)$ have variance $O(1)$). Therefore, the second term is bounded by $\frac{2\lambda_j}{d} + 2d^{-10} \le \frac{3}{\effrank(\Sigma)}$ in absolute value.
    
    Since $\abs{q} \le \abs{1-\tfrac{1}{d}\sum_{i=1}^d \lambda_i g_i^2}$,
    \[
    \E \Brac{g_j^2 q^2} \le 
    \tfrac{1}{d^2}\sum_{i=1}^d \lambda_i^2 \E g_j^2 (1-g_i^2)^2
    \le \tfrac{20}{d^2} \sum_{i=1}^d \lambda_j^2
    \le \tfrac{20}{d}\Norm{\Sigma} 
    \le \tfrac{20}{\effrank(\Sigma)}\,.
    \]
    
    Now let us bound the fourth term. Since $\abs{q} \le 1/10$,
    \[
    \Abs{\sum_{k=3}^{\infty} \E \Brac{g_j^2 q^k}} \le \E \Brac{g_j^2 q^2} 
    \le \frac{20} {\effrank(\Sigma)}\,.
    \]

    Hence for all $i\in[d]$, $\frac{\lambda'_i}{\lambda_i} = 1 + O(\frac{1}{\effrank(\Sigma)})$.
\end{proof}

The following proposition shows that a good estimator of $\Sigma'$ is also a good estimator of $\Sigma$.

\begin{proposition}\label{prop:triangle-inequality}
Let $\Sigma, \Sigma', \hat{\Sigma}\in \R^{d\times d}$ 
be positive definite matrices such that for some $\Delta \le 1$,
\[
\normf{{\Sigma}^{-1/2}\, \Sigma'\, {\Sigma}^{-1/2} - \Id} \le \Delta
\]
and
\[
\normf{\Paren{\Sigma'}^{-1/2}\, \hat{\Sigma}\, \Paren{\Sigma'}^{-1/2} - \Id} \le \Delta\,.
\]
Then
\[
\normf{{\Sigma}^{-1/2}\, \hat{\Sigma}\, {\Sigma}^{-1/2} - \Id} \le 3\Delta\,.
\]
The same is true if Frobenius norm is replaced by spectral norm.
\end{proposition}
\begin{proof}
\begin{align*}
\normf{{\Sigma}^{-1/2}\, \hat{\Sigma}\, {\Sigma}^{-1/2} - \Id} 
&\le
\normf{{\Sigma}^{-1/2}\, \hat{\Sigma}\, {\Sigma}^{-1/2} -{\Sigma}^{-1/2}\, \Sigma'\, {\Sigma}^{-1/2}} +\normf{{\Sigma}^{-1/2}\, \Sigma'\, {\Sigma}^{-1/2} - \Id}
\\&\le \norm{\Sigma^{-1/2}\Paren{\Sigma'}^{1/2}}^2 \cdot \normf{\Paren{\Sigma'}^{-1/2}\, \hat{\Sigma}\, \Paren{\Sigma'}^{-1/2} - \Id} + \Delta
\\&\le \norm{\Sigma^{-1/2}\Paren{\Sigma'}^{1/2}}^2\cdot \Delta + \Delta\,.
\end{align*}
Denote $A = \Sigma^{-1/2}\Paren{\Sigma'}^{1/2}$. Since $\norm{AA^\top - \Id} \le \normf{AA^\top - \Id} \le \Delta$, 
\[
\norm{A}^2 = \norm{AA^\top} \le 1 + \norm{AA^\top - \Id} \le 1 + \Delta \le 2\,.
\]
Hence we get the desired bound. It is clear that this derivation is also correct if we replace Frobenius norm by spectral norm.
\end{proof}

Let us show how to estimate $\Sigma' = \Cov_{y \sim \cN(0,\Sigma)}{\cF(y)}$ up to error $O\Paren{\frac{1}{\log d} + \sqrt{\e}}$ in relative spectral norm. 
 
Let us split the samples into $2$ equal parts, and run 
the covariance estimation algorithm in relative spectral norm from Theorem 1.2 from \cite{KS17} on the projection onto the sphere of the first part. 
We can do that, since $\cF(y)$ is $O(1)$-sub-Gaussian. Indeed, for all $u\in\R^d$,
\[
\Paren{\iprod{u, \cF(y)}^p}^{1/p}\le 1.1\Paren{\iprod{u, y}^p}^{1/p} \cdot C_G\cdot \sqrt{p}\Paren{\E\iprod{u,y}^2}^{1/2} \le  \frac{1.1}{0.9} \cdot C_G\cdot \sqrt{p}\Paren{\E\iprod{u,\cF(y)}^2}^{1/2}\,,
\]
where $C_G \le O(1)$ is the sub-Gaussian parameter for the Gaussian distribution.

With high probability we get an estimator $\hat{\Sigma}_1$ such that 
    \[
    \norm{(\Sigma')^{-1/2}\, \hat{\Sigma}_1\, (\Sigma')^{-1/2} - \Id} \le O\Paren{\sqrt{{\e}}}\,.
    \]
\lemmaref{lem:closeness} and our assumption on the effective rank imply that 
    \[
    \Norm{\Sigma^{-1/2}\,
    \Paren{\Cov_{y \sim \cN(0,\Sigma)}{\cF(y)} - \Sigma}
    \,\Sigma^{-1/2}}  
    \le  O\Paren{\frac{1}{\log(d)}} \,.
    \]
Hence, by \cref{prop:triangle-inequality}, 
$0.99 \cdot \Id \preceq \hat{\Sigma}_1^{-1}\Sigma \preceq 1.01\cdot \Id$. 
If we multiply the second part of the samples by $\hat{\Sigma}_1^{-1}$, we get an elliptical distribution with new scatter matrix $\tilde{\Sigma} = \rho \hat{\Sigma}_1^{-1} \Sigma$, where $\rho = d/\Tr(\hat{\Sigma}_1^{-1} \Sigma)$. $\tilde{\Sigma}$ has trace $d$ and its eigenvalues are between $0.9$ and $1.1$, and, in particular, its effective rank is $\Omega(d)$.

\section{Fourth Moment of the Spatial Sign}
In this section and further in the paper we use a slightly modified definition of the function $\cF$:
\begin{definition}\label{def:function}
    \[
    \cF(x) = 
    \begin{cases}
        \spsign(x) = {\sqrt{d}} \cdot x/{\norm{x}} &\text{if } d-\Delta \le \norm{x}^2 \le d+\Delta\,,\\
        \sqrt{d/(d-\Delta)}\cdot x &\text{if } \norm{x}^2 < d-\Delta\,,\\
        \sqrt{d/(d+\Delta)}\cdot x &\text{if } \norm{x}^2 > d+\Delta\,.
    \end{cases}
    \]    
for $\Delta = C\cdot \sqrt{d\log d}$, where $C$ is a large enough absolute constant.
\end{definition} 

If $\effrank(\Sigma)\ge \Omega(d)$, then with high probability $\cF(y_i)$ coincides with $\spsign(y_i)$ for all $n=\poly(d)$ samples $y_1,\ldots y_n\simiid\cN(0,\Sigma)$. Note that as long as $\effrank(\Sigma) \ge \Omega(d)$, the proof of \cref{lem:closeness} remains the same, so $\Cov \cF(y)$ is $O(1/d)$-close to $\Sigma$ in relative spectral norm.

\begin{lemma}\label{lem:tensor}
    Let $\Sigma\in \R^{d\times d}$ be a positive definite matrix such that $0.9\Id \preceq \Sigma \preceq 1.1\Id$, $\Tr(\Sigma) = d$.
    Let $g\sim \cN(0,\Id)$, $\Sigma'= \Cov(\cF(\Sigma^{1/2}g))$ and $\phi(g)= \Paren{\Sigma'}^{-1/2} \cF(\Sigma^{1/2}g)$.
    
    Let $T = \E\Paren{\phi(g)\transpose{\phi(g)} - \Id}^{\otimes 2}$ and $S = \E_{g\sim N(0,\Id)}\Paren{gg^\top / \norm{g}^2 - \Id}^{\otimes 2}$. Then for all $i_1 i_2 i_3 i_4, i, j \in [d]$, the components of $T$ in the eigenbasis of $\Sigma$ are as follows:
    \begin{enumerate}
        \item  $T_{i_1 i_2 i_3 i_4} = S_{i_1 i_2 i_3 i_4} = 0$ if some of the $i_m$ is not equal to the others.
        \item $T_{iiii} = 2 + O(1/d) = S_{iiii} + O(1/d)$.
        \item For $i\neq j$, $T_{ijij} = 1 + O(1/d) = S_{ijij} + O(1/d)$.
        \item For $i\neq j$, $T_{iijj} = S_{iijj} + O\Paren{{\Norm{\Sigma-\Id}}/{d}} + O(\log^2(d)/d^2)$.
    \end{enumerate}
\end{lemma}
\begin{proof}
 Without loss of generality we assume $\Normf{M}=1$. 
    Let us bound the entries of the tensor 
    $T := \E\Paren{\phi(g)\transpose{\phi(g)} - \Id}^{\otimes 2}$ in the eigenbasis of $\Sigma$:
    \[
    T_{i_1 i_2 i_3 i_4} = 
    \E
    \Paren{\sqrt{\tfrac{\lambda_{i_1}}{\lambda'_{i_1}}\cdot\tfrac{\lambda_{i_2}}{\lambda'_{i_2}}}\cdot g_{i_1}g_{i_2}/R(g)^2 - \Id_{i_1i_2}} 
    \Paren{\sqrt{\tfrac{\lambda_{i_3}}{\lambda'_{i_3}}\cdot\tfrac{\lambda_{i_4}}{\lambda'_{i_4}}}\cdot g_{i_3}g_{i_4}/R(g)^2 - \Id_{i_3i_4} }
     \,,
    \]
    where $R(g)^2 = {\tfrac{1}{d}\sum_{i=1}^d \lambda_i g_i^2 \cdot \ind{\sum_{i=1}^d \lambda_i g_i^2\in[d - \Delta,d + \Delta]} 
    + \tfrac{d - \Delta}{d} \ind{\sum_{i=1}^d \lambda_i g_i^2 < d-\Delta} 
    + \tfrac{d+\Delta}{d} \ind{\sum_{i=1}^d \lambda_i g_i^2 > d+\Delta}
    }$, where $\Delta= O\Paren{\sqrt{d\log d}}$.
    
    First, note that in the case $\lambda_1 = \ldots = \lambda_d = 1$, the components $T_{i_1 i_2 i_3 i_4}$ and $S_{i_1 i_2 i_3 i_4}$ are expectations of random variables that coincide with probability at least $1/\poly(d)$ and have variance bounded by $O(1)$. Hence the difference between them is at most $O(d^{-10})$.
    
    Consider the entries $T_{i_1 i_2 i_3 i_4}$ such that some $i_m$ is not equal to any of the others. 
    In this case, since $\sign(g_{i_m})$ is independent of $\abs{g_{i_m}}$ and of all other $g_j$, $T_{i_1 i_2 i_3 i_4} = 0$.
    
    Consider the terms of the form $T_{iijj}$ (including the case $i=j$). 
    \[
    T_{iijj} = \frac{\lambda_i\lambda_j}{\lambda'_i\lambda'_j}\cdot \E\frac{g_i^2g_j^2}{R(g)^4} - 1 = \Paren{1 + O(1/d)}\cdot \E\frac{g_i^2g_j^2}{R(g)^4} - 1 \,.
    \]
    As in the proof of \lemmaref{lem:closeness}, let $q = 1-R(g)^2$. Since $\abs{2q - q^2} \le 1/2$,
    \[
    \E\frac{g_i^2g_j^2}{R(g)^4} = \E\frac{g_i^2g_j^2}{1-(2q - q^2)} 
    = \E g_i^2g_j^2  + \sum_{k=1}^\infty \E g_i^2g_j^2 (2q - q^2)^k\,.
    \]
    The first term is $1$ if $i\neq j$ and $3$ if $i=j$. Let us bound the other terms.
    By the same argument as in \lemmaref{lem:closeness}, 
    $\E g_i^2g_j^2 q = \frac{2\lambda_j + 2\lambda_i}{d} + O(d^{-10})$.
     If we write $\lambda_k = 1 + \delta_k$, this expression differs from the same expression for $\lambda_k = 1$ by $O(\delta_i / d + \delta_j/d)$.
    Similarly,
\[
    \E \Brac{g_i^2g_j^2 q^2} 
    = \E g_i^2g_j^2\Paren{\tfrac{1}{d}\sum_{m=1}^d\lambda_m (1-g_m^2)}^2+ O(d^{-10})\,.
    \]
    Since $g_i$ are independent and $\E(1-g_m^2)=1$,
    \[
    \E \E g_i^2g_j^2  \Paren{\tfrac{1}{d}\sum_{m=1}^d\lambda_m (1-g_m^2)}^2
    = \tfrac{1}{d^2}\sum_{m=1}^d \lambda_m^2 \E g_i^2g_j^2 (1-g_m^2)^2 + \tfrac{2}{d^2}\lambda_i\lambda_j \E g_i^2g_j^2 (1-g_i^2)(1-g_j^2) \,.
    \]

    If we write $\lambda_k = 1 + \delta_k$ for $\abs{\delta_k}\le 0.1$, this expression differs from the same expression for $\lambda_k = 1$ by at most $O(\max_{k\in[d]} \delta_k / d)$.
     
    Finally, consider
\[
    \E \Brac{g_i^2g_j^2 q^3} 
    = \E g_i^2g_j^2\Paren{\tfrac{1}{d}\sum_{m=1}^d\lambda_m (1-g_m^2)}^3 + O(d^{-10})\,.
    \]
    Note that,
    \begin{align*}
    \E g_i^2g_j^2\Paren{\tfrac{1}{d}\sum_{m=1}^d\lambda_m (1-g_m^2)}^3
    = &\tfrac{1}{d^3} \sum_{m=1}^d \lambda_m^3 \E g_i^2g_j^2 (1-g_m^2)^3 + \tfrac{1}{d^3} \sum_{m=1}^d \lambda_i\lambda_m^2\E g_i^2g_j^2(1-g_i^2)(1-g_m^2)^2 
    \\&+ \tfrac{1}{d^3} \sum_{m=1}^d \lambda_j\lambda_m^2\E g_i^2g_j^2(1-g_j^2)(1-g_m^2)^2
    + O(1/d^3) \,.
    \end{align*}
    Since each term in each sum is bounded by $O(1)$, $\abs{\E \brac{g_i^2g_j^2 q^3} }\le O(1/d^2)$.
    Since $q^4 \le O(\log^2(d)/d^2)$, the sum of the terms of the series $\sum_{k=1}^\infty \E g_i^2g_j^2 (2q - q^2)^k$ that are not linear in $q$, $q^2$, or $q^3$, is bounded by $O(\log^2(d)/d^2)$. Hence we get the desired expressions for $T_{iijj}$. Note that $T_{ijij}$ for $i\neq j$ is equal to $1+T_{iijj}$, so $T_{ijij} = 1 + O(1/d)$.

\end{proof}

\section{Generalized Stability}

\begin{lemma}\label{lem:stability}
    Let $m,\e,\delta, r, \cV, \cP, \mu, Q$ be as in \cref{def:stability}, and suppose that $\e < \delta < 0.1$. Let $S$ be a $\Paren{2\e,\delta,r,\cV,\cP}$-stable multiset of size $n\in \N$ with respect to $\mu$ and $Q$. Let $T = \Set{x_1,\ldots,x_n}$ be an $\e$-corruption of $S$. 

    Let $\cW_\e \subset \R^{n}$ be the set of vectors whose entries satisfy $0 \le w_i\le \frac{1}{n}$, $\sum_{i=1}^n w_i \ge 1 -\e$. For $w\in \cW_\e$ denote $\mu_w = \frac{1}{\sum_{i=1}^n w_i} \sum_{i=1}^{n} w_i x_i$ and $\Sigma_w = \frac{1}{\sum_{i=1}^n w_i}\sum_{i=1}^n w_i\paren{x_i-\mu_w}\paren{x_i-\mu_w}^\top$. 

    If for some $w\in \cW_\e$ and $\lambda > 0$, for all $P\in \cP$, $\Abs{\iprod{P, Q-\Sigma_w}} \le \lambda$, then 
    \[
    \sup_{v\in \cV}\; \Abs{\iprod{v, \mu_w-\mu}} \le O\Paren{\delta + \sqrt{\e\lambda} + \e r}\,.
    \]
\end{lemma}

\begin{proof}
Let $G$ be the set of indices of uncorrupted points in $T$, and let $B = [n]\setminus G$.
Note that
\[
\sum_{i=1}^{n} w_i\iprod{v, x_i - \mu} = \Paren{\tfrac{1}{\card{G}}\sum_{i\in G} \iprod{v,x_i - \mu}}
+ \sum_{i\in G} \Paren{w_i - \frac{1}{\card{G}}} \iprod{v,x_i - \mu} 
+ \sum_{i\in B} w_i \iprod{v,x_i - \mu}. 
\]
The first term is bounded by $\delta$. Let us bound the second term.
By \CS inequality, the second term can be bounded as follows:
\begin{align*}
\Paren{\sum_{i\in G} \Paren{w_i - \frac{1}{\card{G}}} \iprod{v,x_i - \mu}}^2
&\le
\Paren{\sum_{i\in G} \Paren{\frac{1}{\card{G}} - w_i}}
\cdot \sum_{i\in G}  \Paren{\frac{1}{\card{G}} - w_i} \iprod{v,x_i - \mu}^2 
\\&\le \e \cdot \sum_{i\in G}{\frac{1}{\card{G}}} \iprod{v,x_i - \mu}^2 
\\&\le \e \cdot \Paren{vQv^\top + \delta^2/\e}
\\&\le \e r^2 + \delta^2\,.
\end{align*}

It remains to bound the third term. By \CS inequality,
\[
\Paren{\sum_{i\in B} w_i \iprod{v,x_i - \mu}}^2\le
\Paren{\sum_{i\in B} w_i}
\cdot \sum_{i\in B}  w_i \iprod{v,x_i - \mu}^2 
\le
\e \cdot \sum_{i\in B}  w_i \iprod{v,x_i - \mu}^2 \,.
\]
Let $w'_i = w_i$ for all $i\in G$ and $w_i'=0$ otherwise. It follows that
\[
\Sigma_{w'} = \sum_{i\in G} w_i'(x_i-\mu)(x_i-\mu)^\top - \Paren{\sum_{i=1}^n w_i'} (\mu-\mu_{w'})(\mu-\mu_{w'})^\top\,.
\]
Note that
\[
\abs{\iprod{v,\mu-\mu_{w'}}} \le \abs{\iprod{v,\mu-\tfrac{1}{\card{G}}\sum_{i\in G} x_i}} + \abs{\iprod{v,\sum_{i\in G} \Paren{\tfrac{1}{\card{G}}-w_i} \paren{x_i-\mu}}} 
\le \delta + \sqrt{\e r^2 + \delta^2}\,.
\]
In addition, note that
\[
\sum_{i\in G} w_i'\iprod{v, x_i-\mu}^2 = \tfrac{1}{\card{G}}\sum_{i\in G} \iprod{v, x_i-\mu}^2 - 
\sum_{i\in G} \Paren{\tfrac{1}{\card{G}}-w_i'}\iprod{v, x_i-\mu}^2 = vQv^\top - \delta^2 - \e r^2\,.
\]
Hence
\[
v\Sigma_{w'}v^\top = vQv^\top - O\Paren{\delta^2 + \e r^2}\,.
\]

Since
\[
\sum_{i\in G} w_i\paren{x_i-\mu}\paren{x_i-\mu}^\top \succeq \Sigma_{w'}\,,
\]
\[
\sum_{i\in G} w_i \iprod{v, x_i-\mu}^2 \ge vQv^\top - O(\delta^2 + \e r^2).
\]
Therefore,
\begin{align*}
\sum_{i\in B} w_i \iprod{v, x_i-\mu}^2 
&\le
\sum_{i=1}^n w_i \iprod{v, x_i-\mu}^2 - \sum_{i\in G} w_i \iprod{v, x_i-\mu}^2
\\&\le \lambda + vQv^\top - vQv^\top + O(\delta^2 + \e r^2)
\\&\le \lambda +  O(\delta^2 + \e r^2)\,.
\end{align*}

By \CS inequality,
\[
\abs{\sum_{i\in G} w_i \iprod{v,x_i - \mu}} \le \sqrt{\e} \cdot \sqrt{\lambda +  O(\delta^2 + \e r^2)}\,.
\]
Since $\sum_{i=1}^n w_i \ge 1/2$, we get the desired bound.

\end{proof}

To analyze the (generalized) filtering algorithm, we need the following notion.

\begin{definition}[Adequacy of $\cP$]\label{def:adequate}
    Let $m\in \N$, $\cV \subset \R^m$, $\cP \subset \R^{m\times m}$ such that $\cV\otimes \cV \subset \cP$. We say that $\cP$ is \emph{adequate with resect to $\cV$}, if
    \begin{enumerate}
        \item For all $a \in \R^m$ and $P \in \cP$, $\Iprod{a^{\otimes 2}, P} \ge 0$
        \item For all $a,b\in \R^m$ and $P\in \cP$, 
        \[
        \Iprod{a\otimes b, P}^2 \le \Paren{\sup_{v\in \cV}\Iprod{a, v}^2}\cdot\Paren{\sup_{v\in \cV}\Iprod{b, v}^2}\,.
        \]
    \end{enumerate}
\end{definition}

When $\cP = \cV\otimes \cV$, it is clearly adequate with respect to $\cV$. Further we show (\cref{lem:adequate}) that the \emph{sum-of-squares relaxation} of the set of tensors of the form $v^{\otimes4}$ is adequate with respect to the set of matrices of the form $v^{\otimes 2}$. 

\subsection{Filtering Algorithm}
The following algorithm is a generalized version of filtering. Note that it updates the samples via some (known) transformation $F$. When we use if for estimation in relative spectral norm, $F$ is a linear transformation that maps the points into the isotropic position so that $\mu^{(t)}$ is the identity matrix. When we use it for estimation in Frobenius norm, $F$ is simply the identity transformation.
 
\begin{mdframed}[linecolor=black, linewidth=0.8pt, roundcorner=4pt,
  innertopmargin=0.8em, innerbottommargin=0.8em, innerleftmargin=1em, innerrightmargin=1em]
\begin{algorithm}[Filtering]\label{alg:filtering}
\begin{algorithmic}
\State
\State \textbf{Input:} $x\in \R^{n\times m}, Q\in \R^{m\times m}, \cP\subset \R^{m\times m}, F:\R^n\times \R^{n\times m} \to \R^{n\times m}, \varepsilon, \delta\in (0,1), r,C>0$.
\State For all $i\in [n]$, let $w^{(0)}_i \gets 1/n$
\State Let $x^{(0)} \gets x$, $\mu^{(0)} \gets \sum_{i=1}^n \tfrac{1}{n} x_i$, $\Sigma^{(0)} = \sum_{i=1}^n \tfrac{1}{n} \paren{x_i - \mu^{(0)}}\paren{x_i - \mu^{(0)}}^\top$.
\State Let $t \gets 0$.
\While {there exists $P^*\in \cP$ such that $\Abs{\iprod{P^*, Q-\Sigma^{(t)}  }} > C \cdot \Paren{\delta^2/\e + \e r^2}$}
    \State For all $i\in [n]$, let 
    $\tau_i^{(t)} = {\Iprod{P^*, \paren{x_i^{(t)}-\mu^{(t)}}\paren{x_i^{(t)}-\mu^{(t)}}^\top } }$
    \State Sort the $\tau^{(t)}_i$ in decreasing order.
    \State Let $N$ be the first index such that $\sum_{i=1}^{N} w^{(t)}_i > 2\varepsilon$.
    \State Let $w^{(t+1)}_i = \Paren{1-\tau_i/\tau_{\max}}w_i^{(t)}$ for $i \le N$, and $w^{(t+1)}_i = w^{(t)}_i$ otherwise.
    \State Let $x^{(t+1)} \gets F(w^{(t+1)}, x^{(0)})$.
    \State Let $\mu^{(t + 1)} \gets \frac{1}{\sum_{i=1}^n w_i^{(t+1)}}\sum_{i=1}^n w_i^{(t+1)} x_i^{(t+1)}$.
    \State Let $\Sigma^{(t+1)} \gets \frac{1}{\sum_{i=1}^n w_i^{(t+1)}}\sum_{i=1}^n w_i^{(t+1)}\paren{x_i^{(t+1)}-\mu^{(t+1)}}\paren{x_i^{(t+1)}-\mu^{(t+1)}}^\top$.
    \State Update $t \gets t + 1$.
\EndWhile
\State \Return $\mu^{(t)}$
\end{algorithmic}
\end{algorithm}
\end{mdframed}

\begin{theorem}\label{thm:filtering}
    Let $C > 0$ be a sufficiently large universal constant.
    Let $m,\e,{\delta}, r, \cV, \cP, \mu, Q$ be as in \cref{def:stability}, where $\e < {\delta} < 0.01$ and $\cP$ is adequate with respect to $\cV$.
    Suppose that at each step of \cref{alg:filtering} there exists 
    \[
    \delta^{(t)} \le O\Paren{\delta + \min\Set{\sqrt{\e}, \e^{3/4}\Abs{\iprod{P^*, Q-\Sigma^{(t)}  }} ^{1/4}} + \e r }
    \]
    such that
    $\set{x_1^{(t)},\ldots x_n^{(t)}}$ is an $\e$-corruption of a $\Paren{10\e,\delta^{(t)},r,\cV,\cP}$-stable multiset of size $n$ with respect to $\mu^{(t)}$ and $Q$.
    Then the while-statement of \cref{alg:filtering} is checked at most $O(n)$ times, and the algorithm outputs $\hat{\mu}$ such that 
    \[
    \sup_{v\in \cV}\; \Abs{\iprod{v, \hat{\mu}-\mu}} \le O\Paren{\delta + \e r}\,.
    \]
\end{theorem}

The following lemma easily implies \cref{thm:filtering}:

\begin{lemma}[Invariant of the Filtering Algorithm]\label{lem:filtering-invariant} 
Let $G\subseteq [n]$ be the set of indices of uncurrupted samples, and $B = [n]\setminus G$.
Suppose that in iteration $t$, \begin{equation}\label{eq:filtering-invariant}
\sum_{i \in G} \tfrac{1}{n} - w_i^{(t)} \quad \text{<} \quad \sum_{i \in B} \tfrac{1}{n} - w_i^{(t)}\,,
\end{equation}
and furthermore there exists $P^*\in \cP$ such that
\[
\Abs{\iprod{P^*, Q-\Sigma_{w^{(t)}}}} > C\cdot \Paren{\delta^2/\e + \e r^2}.
\]
Then
\[
\sum_{i \in G} \left( \tfrac{1}{n} - w_i^{(t+1)} \right) <
\sum_{i \in B} \left( \tfrac{1}{n} - w_i^{(t+1)} \right).
\]
\end{lemma}

\begin{proof}[Proof of \cref{thm:filtering}]
Let us show that the algorithm terminates after at most $t' = \lceil 2\e n\rceil $ iterations. Assume towards a contradiction that it does not terminate after $t'$ iterations. The number of nonzero entries of $w^{(t)}$ decreases by at least $1$ at each iteration. Hence we have at least $\e n$ zero entries of $w^{(t')}$ that are in $G$. By \cref{lem:filtering-invariant},
\[
\e \le \sum_{i\in G} \Paren{\tfrac{1}{n} - w^{(t')}_i} <  \sum_{i\in G} \Paren{\tfrac{1}{n} - w^{(t')}_i} \le \frac{\card{B}}{n} \le \e\,,
\]
which is a contradiction.

By \cref{lem:stability}, we get the desired bound. 
\end{proof}

Now let us prove \cref{lem:filtering-invariant}.

\begin{proof}
For simplicity let us denote $w = w^{(t)}$ and $w' = w^{(t+1)}$. Note that it is enough to show that
\[
\sum_{i\in G} \Paren{w_i - w'_i} < \sum_{i\in B} \Paren{w_i - w'_i}\,.
\]

Since for $i\le N$, $w_i - w'_i = \tfrac{1}{\tau_{\max}}\tau_i w_i$ and for $i > N$, $w_i - w'_i = 0$, it is equivalent to
\[
\sum_{i\in G\cap [N]} w_i \tau_i < \sum_{i\in B \cap [N]} w_i\tau_i\,.
\]

We will show it by proving the following two lemmas:

\begin{lemma}[Upper bound for the good samples]\label{lem:filtering-upper-bound}
\[
\sum_{i\in G\cap [N]} w_i \tau_i \le O\Paren{\paren{\delta^{(t)}}^2/\e + \e r^2  + \e^2 \Abs{\Iprod{P^*, Q-\Sigma_{w}}} }\,.
\]
\end{lemma}

\begin{lemma}[Lower bound for the bad samples]\label{lem:filtering-lower-bound}
\[
\sum_{i\in B\cap [N]} w_i \tau_i \ge \Abs{\Iprod{P^*, Q-\Sigma_{w}}} / 2\,.
\]
\end{lemma}

\begin{proof}[Proof of \cref{lem:filtering-upper-bound}]
\begin{align*}
\sum_{i\in G\cap [N]} w_i \tau_i 
&= \sum_{i\in G\cap [N]} w_i \Iprod{P^*, \Paren{x_i-\mu_{w^{}}}\Paren{x_i-\mu_{w^{}}}^\top }
\\&\le O\Paren{\sum_{i\in G\cap [N]} w_i {\Iprod{P^*, \Paren{x_i-\mu}\Paren{x_i-\mu}^\top }}
+ \sum_{i\in G\cap [N]} w_i {\Iprod{P^*, \Paren{\mu-\mu_{w^{}}}\Paren{\mu-\mu_{w^{}}}^\top }} }
\\&\le O\Paren{\sum_{i\in G\cap [N]} \tfrac{1}{n} {\Iprod{P^*, \Paren{x_i-\mu}\Paren{x_i-\mu}^\top }}
+ \e {\Iprod{P^*, \Paren{\mu-\mu_{w^{}}}\Paren{\mu-\mu_{w^{}}}^\top }} }\,.
\end{align*}
The first term can be bounded using stability:
\begin{align*}
\sum_{i\in G\cap [N]} \tfrac{1}{n} {\Iprod{P^*, \Paren{x_i-\mu}\Paren{x_i-\mu}^\top }}
&= \sum_{i\in G}  \tfrac{1}{n}\Iprod{P^*, \Paren{x_i-\mu}\Paren{x_i-\mu}^\top } - \sum_{i\in G\setminus [N]}  \tfrac{1}{n}\Iprod{P^*, \Paren{x_i-\mu}\Paren{x_i-\mu}^\top}
\\&\le \tfrac{\card{G}}{n}\Iprod{P,Q} - \tfrac{\card{G}\setminus[N]}{n}\Iprod{P,Q} + O\Paren{\paren{\delta^{(t)}}^2/\e}
\\&\le O\Paren{\e r^2 + \paren{\delta^{(t)}}^2/\e}\,.
\end{align*}
Let us bound the second term. Since $\cP$ is adequate with respect to $\cV$,
\begin{align*}   
\Iprod{P^*, \Paren{\mu-\mu_{w^{}}}\Paren{\mu-\mu_{w^{}}}^\top }
&\le O\Paren{\max_{v\in\cV}\Iprod{v, \mu-\mu_{w^{}}}^2}
\\&\le O\Paren{\paren{\delta^{(t)}}^2 + \e \Abs{\Iprod{P^*, Q-\Sigma_{w}}} + \e^2 r^2}\,,
\end{align*}
where the second inequality follows from \cref{lem:stability}.
\end{proof}

\begin{proof}[Proof of \cref{lem:filtering-lower-bound}]
    
\end{proof}
First, we show that 
\[
\sum_{i\in B} w_i \tau_i \ge 0.9\cdot {\Abs{\Iprod{P^*, Q-\Sigma_{w}}}}\,.
\]
Note that
\[
\sum_{i=1}^n w_i\tau_i = \sum_{i=1}^n w_i \Iprod{P^*, \Paren{x_i-\mu_{w^{}}}\Paren{x_i-\mu_{w^{}}}^\top } = \Iprod{P^*, \Sigma_w - Q} + \Iprod{P^*, Q}\,.
\]
Consider
\begin{align*}
\Abs{\sum_{i\in G} w_i\tau_i - \Iprod{P^*, Q}} 
&=  \Abs{\sum_{i\in G} w_i \Iprod{P^*, \Paren{x_i-\mu_{w^{}}}\Paren{x_i-\mu_{w^{}}}^\top } -  \Iprod{P^*, Q}}
\\&\le \Abs{\sum_{i\in G} w_i \Iprod{P^*, \Paren{x_i-\mu}\Paren{x_i-\mu}^\top } -  \Iprod{P^*, Q}} + 
\Abs{\sum_{i\in G} w_i {\Iprod{P^*, \Paren{\mu-\mu_{w^{}}}\Paren{\mu-\mu_{w^{}}}^\top }}}
\\& \qquad + \Abs{\sum_{i\in G} w_i {\Iprod{P^*, \Paren{x_i-\mu}\Paren{\mu-\mu_{w^{}}}^\top }}}
+\Abs{\sum_{i\in G} w_i {\Iprod{P^*, \Paren{\mu-\mu_{w^{}}}\Paren{x_i-\mu}^\top }}}\,.
\end{align*}

The first term is bounded by $O(\paren{\delta^{(t)}}^2/\e + \e r^2)$ (by stability). The second term is bounded by $O\Paren{\paren{\delta^{(t)}}^2 + \e \Abs{\Iprod{P^*, Q-\Sigma_{w}}} + \e^2 r^2}$, as in the proof of \cref{lem:filtering-upper-bound}. The other terms are also bounded by $O\Paren{\paren{\delta^{(t)}}^2 + \e \Abs{\Iprod{P^*, Q-\Sigma_{w}}} + \e^2 r^2}$ (by adequacy of $\cP$, \cref{lem:stability} for the bound on $\max_{v\in \cV}\iprod{v, \mu - \mu_w}$ and stability for the bound on $\max_{v\in \cV}\iprod{v, \mu - \sum_{i\in G} w_ix_i}^2$).

Hence
\[
\sum_{i\in B} w_i \tau_i \ge \Iprod{P^*, \Sigma_w - Q} - \Abs{\sum_{i\in G} w_i\tau_i - \Iprod{P^*, Q}} \ge \Paren{1-\e}\Iprod{P^*, \Sigma_w - Q} - O\Paren{\paren{\delta^{(t)}}^2/\e + \e r^2}\,.
\]

By the bound on $\delta^{(t)}$ and since $\Abs{\iprod{P^*, Q-\Sigma_{w^{(t)}}}} > C\cdot \Paren{\delta^2/\e + \e r^2}$, we get that $\sum_{i\in B} w_i \tau_i \ge 0.99\cdot {\Abs{\Iprod{P^*, Q-\Sigma_{w}}}}$. Let us bound $\sum_{i\in B\setminus[N]} w_i \tau_i$. By \cref{lem:filtering-upper-bound}, for each $i > N$,
\[
\tau_i \le \frac{1}{\sum_{i\in G\cap [N]}w_i} \sum_{i\in G\cap N} w_i \tau_i \le O\Paren{\paren{\delta^{(t)}}^2/\e^2 + r^2  + \e \Abs{\Iprod{P^*, Q-\Sigma_{w}}} }\,,
\]
where we used the fact that $\tau_i \le \tau_j$ for all $j\in G\cap[N]$. Hence
\[
\sum_{i\in B\setminus[N]} w_i \tau_i \le O\Paren{\paren{\delta^{(t)}}^2/\e + \e r^2  + \e^2 \Abs{\Iprod{P^*, Q-\Sigma_{w}}} }\le  0.1 \Abs{\Iprod{P^*, Q-\Sigma_{w}}}\,.
\]
Therefore,
\[
\sum_{i\in B\cap [N]} w_i \tau_i \ge \Abs{\Iprod{P^*, Q-\Sigma_{w}}} / 2\,.
\]
\end{proof}
\cref{lem:filtering-invariant} follows from \cref{lem:filtering-upper-bound} and \cref{lem:filtering-lower-bound} due to the bound on $\delta^{(t)}$ and since $\Abs{\iprod{P^*, Q-\Sigma_{w^{(t)}}}} > C\cdot \Paren{\delta^2/\e + \e r^2}$.

\section{Spectral Covariance Stability}

\begin{lemma}[From Euclidean to Spectral Stability]\label{lem:stability-frob-to-spectral}
    Let $0 < \e < \delta < 0.1$, $r, R > 0$, 
    $\cB_R = \set{V\in \R^{d\times d} : \normf{V} \le R}$. Suppose that $\set{X_1,\ldots,X_n} \subset \R^{d\times d}$ is an $\Paren{\e,\delta,r,\cB_R,\cB_R\otimes \cB_R}$-stable set with respect to some $M\in \R^{d\times d}$ and $Q\in \R^{d^2\times d^2}$.
    Let $\cS_R = \set{uu^\top \in \R^{d\times d} : \norm{u}\le \sqrt{R}}$ and let 
    \[
    \cP_R = \Set{\pE v^{\otimes 4} : v\in \R^d, \pE \text{ is a degree-$4$ pseudo-expectation that satisfies the constraint } \norm{v}^2 \le R}\,.
    \]
    
    Then for each $Q'\in \R^{d^2\times d^2}$ such that 
    \[
    \Abs{\Iprod{P, Q-Q'}}\le O(\delta^2/\e)\,,
    \]
    $\set{X_1,\ldots,X_n}$ is an $\Paren{\e,\delta,r,\cS_R,\cP_R}$-stable set with respect to $M$ and $Q'$. 
\end{lemma}
\begin{proof}
    Let $\cC_R = \conv\Paren{\cB_R\otimes \cB_R}$. Note that it is the set of $d^2\times d^2$ matrices with trace at most $R^2$. Note that $\cS_R \subset \cB_R$, and $\cP_R \subset \cC_R$, since $\sum_{i=1}^d\sum_{j=1}^d \pE u_iu_j \cdot u_iu_j = \pE \sum_{i=1}^d u_i^2 \cdot \sum_{j=1}^d u_j^2 \le R^2$.
    
    Since maximum and minimum of a linear function over a convex set is achieved in its extreme poitns,  $\set{X_1,\ldots,X_n}$ is an $\Paren{\e,\delta,r,\cB_R,\cC_R}$-stable set with respect to $M$ and $Q$, and hence of an $\Paren{\e,\delta,r,\cS_R,\cP_R}$-stable with respect to $M$ and $Q$. 
    
    Finally, it is clear that $Q$ can be replaced by any $Q'$ such that $\Abs{\Iprod{P, Q-Q'}}\le O(\delta^2/\e)$.
\end{proof}

\begin{corollary}\label{cor:spectral-stability}
Let $\Sigma \in \R^{d\times d}$ be a positive definite matrix such that $0.9 \cdot \Id_d \preceq \Sigma \preceq 1.1\cdot \Id_d$. Let $y_1,\ldots,y_n\simiid N(0,\Sigma)$, $\Sigma' = \Cov_{y\sim N(0,\Sigma)} \cF(y)$, and let $\zeta_i = \Paren{\Sigma'}^{-1/2}\cF(y_i)$. Let $R \le O(1)$.

If $\e\log(1/\e) \gtrsim 1/d$ and $n\gtrsim d^2\log^5(d)/\e^2$, then $\zeta_1\zeta_1^\top,\ldots,\zeta_n\zeta_n^\top$ is an $\Paren{\e,\delta,r,\cS_R, \cP_R}$-stable set with respect to $M=\Id_d$ and $Q' = 2\Id_{d^2}$, for some $\e\le \delta \le O(\e{\log(1/\e)})$ and $r \le O(1)$.
\end{corollary}
\begin{proof}
By \ref{lem:frobenius-stability-isotropic}, with high probability, $\zeta_1\zeta_1^\top,\ldots,\zeta_n\zeta_n^\top$ is an $\Paren{\e,\delta,r,\cB_R, \cB_R\otimes \cB_R}$-stable set with respect to $M=\Id_d$ and $Q = \E_{y\sim N(0,\Sigma)} \Paren{\Paren{\Sigma'}^{-1/2}\cF(y)\cF(y)^\top\Paren{\Sigma'}^{-1/2} - \Id}^{\otimes 2}$. By \cref{lem:stability-frob-to-spectral}, it is enough to show that 
\[
\Abs{\Iprod{\pE u^{\otimes 4}, Q} - 2} \le O(\delta^2/\e)\,.
\]

By \cref{lem:tensor}, $Q_{iiii} = 2 \pm O(1/d)$, and for $i\neq j$, $Q_{ijij} = 1 \pm O(1/d)$, $Q_{iijj}= \pm O(1/d)$, and other entries are $0$. Hence
\begin{align*}
\Iprod{\pE u^{\otimes 4}, Q} 
&= \sum_{1\le i,j \le d} \pE u_i^2u_j^2 Q_{ijij} + \sum_{1\le i\neq j \le d} \pE u_i^2u_j^2 Q_{iijj} = 2 \pm O(1/d) = 2 \pm O(\e\log(1/\e))\,.
\end{align*}
\end{proof}

\begin{lemma}\label{lem:stability-linear-transformation}
Let $R \le O(1)$. Let $\zeta_1\zeta_1^\top,\ldots,\zeta_n\zeta_n^\top$ be an $\Paren{\e,\delta,r,\cS_R, \cP_R}$-stable set with respect to $\Id_d$ and $2\Id_{d^2}$, for some $\e\le \delta$ and $r\le O(1)$.

Then for each matrix $A\in \R^{d\times d}$ such that $\norm{A}^2\le R$, $A\zeta_1\Paren{A\zeta_1}^\top,\ldots,A\zeta_n\Paren{A\zeta_n}^\top$ is an $\Paren{\e,\delta',r,\cS_1, \cP_1}$-stable set with respect to $M = AA^\top$ and $Q = 2\Id_{d^2}$,
where $r\le O(1)$ and
\[
\delta' = O\Paren{\delta + \sqrt{\e\norm{AA^\top - \Id_d}}}\,.
\]
\end{lemma}
\begin{proof}
Let us check the first property of stability. 
For each unit $u\in \R^{d}$,
\[
\Abs{\tfrac{1}{\card{S'}} \sum_{x\in S'} \iprod{uu^\top, A\zeta_i\zeta_i^\top A^\top-AA^\top}}
= \Abs{\tfrac{1}{\card{S'}} \sum_{x\in S'} \iprod{A^\top uu^\top A, \zeta_i\zeta_i^\top-\Id_d}}
\le O(\delta)\,,
\]
where the inequality follows from the stability of $\zeta_1\zeta_1^\top, \ldots, \zeta_n\zeta_n^\top$.
Let us check the second property. We need to show that
\[
\Abs{\tfrac{1}{\card{S'}} \sum_{x\in S'} \Iprod{\pE v^{\otimes 4}, \paren{A\zeta_i\zeta_i^\top A^\top-AA^\top}^{\otimes 2}} - \Iprod{\pE v^{\otimes 4}, 2\Id_{d^2}}}\le \paren{\delta'}^2/\e\,.
\]
Note that
\begin{align*}
\Iprod {\pE v^{\otimes 4}, \paren{A\zeta_i\zeta_i^\top A^\top-AA^\top}^{\otimes 2}}
&= \pE \Iprod{vv^\top, A\zeta_i\zeta_i^\top A^\top-AA^\top}^2
\\&= \pE \Iprod{\Paren{A^\top v}\Paren{A^\top v}, \zeta_i\zeta_i^\top -\Id_d}^2
\\&= \Iprod{\pE \Paren{A^\top v}^{\otimes 4}, \zeta_i\zeta_i^\top -\Id_d}\,.
\end{align*}
By the stability of $\zeta_1\zeta_1^\top, \ldots, \zeta_n\zeta_n^\top$, 
\[
\Abs{\tfrac{1}{\card{S'}} \sum_{x\in S'} 
\iprod{\pE \Paren{A^\top v}^{\otimes 4}, \zeta_i\zeta_i^\top -\Id_d}
- \Iprod{\pE \Paren{A^\top v}^{\otimes 4}, 2\Id_{d^2}}}\le O\Paren{\delta^2/\e}\,.
\]
Note that
\begin{align*}
\Abs{\Iprod{\pE \Paren{A^\top v}^{\otimes 4} - \pE v^{\otimes 4}, 2\Id_{d^2}}}
&= 2\Abs{\pE \norm{A^\top v}^4 - \pE \norm{v}^4}
\\&= 2\Abs{\pE \Paren{v^\top A A^\top v}^2 - \pE \Paren{v^\top v}^2}
\\&= 
2\Abs{\pE\Paren{{v^\top \Paren{A A^\top - \Id_d} v}}\cdot \Paren{{v^\top \Paren{A A^\top + \Id_d} v}}}
\\&\le \sqrt{\pE\Paren{{v^\top \Paren{A A^\top - \Id_d} v}}^2 \cdot \pE\Paren{{v^\top \Paren{A A^\top + \Id_d} v}}^2}
\\&\le O\Paren{\norm{A A^\top - \Id_d}}\,,
\end{align*}
where we used the \CS inequality for pseudo-distributions, \cref{fact:spectral-certificates} and $\norm{A}\le O(1)$ . Hence the second property is satisfied. The third property is also obviously satisfied.
\end{proof}

\begin{lemma}\label{lem:non-decreasing-covariance}
Let $\zeta_1\zeta_1^\top,\ldots,\zeta_n\zeta_n^\top$ be an $\Paren{10\e,\delta,r,\cS_1, \cP_1}$-stable set with respect to $\Id_d$ and $Q\in \R^{d^2\times d^2}$, where $\e \le \delta \lesssim 1$. Let $\set{z_1,\ldots,z_n}$ be an $\e$-corruption of $\set{\Paren{\paren{\Sigma'}^{1/2}\zeta_1}\Paren{\paren{\Sigma'}^{1/2}\zeta_1}^\top,\ldots, \Paren{\paren{\Sigma'}^{1/2}\zeta_n}\Paren{\paren{\Sigma'}^{1/2}\zeta_n}^\top}$. Let $w\in \cW_\e$, and let 
$\Sigma'' = \frac{1}{\sum_{i=1}^n w_i}\sum_{i=1}^n w_i z_i z_i^\top$. 
Then 
\[
\norm{\Paren{\Sigma''}^{-1/2}\Paren{\Sigma'}^{1/2}}\le O(1)\,.
\]
\end{lemma}
\begin{proof}
Note that it is enough to show that $\Sigma'' \succeq \Paren{1-O(\delta)} \Sigma'$.
Let $G$ be the set of indices where $z$ coincides with $\paren{\Sigma'}^{1/2}\zeta$
Note that
\[
\Sigma'' \succeq \frac{1}{\sum_{i=1}^nw_i}\sum_{i\in G} w_i \paren{\Sigma'}^{1/2}\zeta_i \zeta_i^\top \paren{\Sigma'}^{1/2} \succeq \Paren{1-O(\e)}\sum_{i\in G} w_i \paren{\Sigma'}^{1/2}\zeta_i \zeta_i^\top \paren{\Sigma'}^{1/2}\,.
\]
Let $u\in \R^d$ be an arbitrary unit vector. Since the function 
\[
u^\top\Paren{\Sigma'' - \Paren{1-O(\e)}\paren{\Sigma'}^{1/2}\Paren{\sum_{i\in G} w_i \zeta_i \zeta_i^\top} \paren{\Sigma'}^{1/2}}u
\]
is linear in $w$, its maximum is achieved in one of the extreme points of $\cW_\e$, i.e. in some set $S_u$ of size $\card{S_u}\ge \Paren{1-\e} n$. Denote $G_{u} = G\cap S_u$. By stability,
\[
\Norm{\frac{1}{\card{G_u}}\sum_{i\in G_u} \zeta_i \zeta_i^\top - \Id_d} \le \delta\,.
\]
It follows that
\begin{align*}
0
&\le
u^\top\Paren{\Sigma'' - \Paren{1-O(\e)}\paren{\Sigma'}^{1/2}\Paren{\sum_{i\in G} w_i \zeta_i \zeta_i^\top} \paren{\Sigma'}^{1/2}}u 
\\&\le
u^\top\Paren{\Sigma'' - \Paren{1-O(\e)}\paren{\Sigma'}^{1/2}\Paren{\sum_{i\in G_u} \tfrac{1}{n} \zeta_i \zeta_i^\top} \paren{\Sigma'}^{1/2}}u
\\&\le
u^\top\Paren{\Sigma'' - \Paren{1-O(\e)}\paren{\Sigma'}^{1/2}\Paren{\sum_{i\in G_u} \tfrac{1}{\card{G_u}} \zeta_i \zeta_i^\top} \paren{\Sigma'}^{1/2}}u
\\&\le
u^\top\Paren{\Sigma'' - \Paren{1-O(\e)}\Paren{1-\delta}\paren{\Sigma'}^{1/2}\paren{\Sigma'}^{1/2}}u
\\&\le u^\top\Paren{\Sigma'' - \Paren{1-O(\delta)}\Sigma'}u
\end{align*}
Hence $\Sigma'' \succeq \Paren{1-O(\delta)} \Sigma'$.
\end{proof}

\begin{lemma}
    Let $R\le O(1)$. Let $\zeta_1\zeta_1^\top,\ldots,\zeta_n\zeta_n^\top$ be an $\Paren{10\e,\delta,r,\cS_R, \cP_R}$-stable set with respect to $\Id_d$ and $2\Id_{d^2}$, where $\e \le \delta \lesssim 1$ and $r\le O(1)$. 
    
    Let $\set{z_1,\ldots,z_n}$ be an $\e$-corruption of $\set{\Paren{\paren{\Sigma'}^{1/2}\zeta_1}\Paren{\paren{\Sigma'}^{1/2}\zeta_1}^\top,\ldots, \Paren{\paren{\Sigma'}^{1/2}\zeta_n}\Paren{\paren{\Sigma'}^{1/2}\zeta_n}^\top}$. Let $w\in \cW_\e$. 
    Denote $\Sigma'' = \frac{1}{\sum_{i=1}^n w_i}\sum_{i=1}^n w_i z_i z_i^\top$ and
    $C = \frac{1}{\sum_{i=1}^n w_i}\sum_{i=1}^n w_i \Paren{\Paren{\Sigma''}^{-1/2}z_iz_i^\top \Paren{\Sigma''}^{-1/2} - \Id}^{\otimes 2}$. Suppose that 
    \[
    \max_{P\in\cP_1}\Abs{\Iprod{C - Q,P}}\le \lambda\,.
    \]
    
Denote $A = \Paren{\Sigma''}^{-1/2}\Paren{\Sigma'}^{1/2}$.
Then $A\zeta_1\Paren{A\zeta_1}^\top,\ldots,A\zeta_n\Paren{A\zeta_n}^\top$ is an $\Paren{\e,\delta',r,\cS_1, \cP_1}$-stable set with respect to $M = AA^\top$ and $Q = 2\Id_{d^2}$,
where $r\le O(1)$ and
\[
\delta' = O\Paren{\delta + \e^{3/4}\lambda^{1/4}}\,.
\]
\end{lemma}
\begin{proof}
By \cref{lem:stability-linear-transformation} and \cref{lem:non-decreasing-covariance} $A\zeta_1\Paren{A\zeta_1}^\top,\ldots,A\zeta_n\Paren{A\zeta_n}^\top$ is an $\Paren{\e,O\Paren{\delta + \sqrt{\e}},r,\cS_1, \cP_1}$-stable set with respect to $M = AA^\top$ and $Q = 2\Id_{d^2}$. By \cref{lem:stability},
\[
\norm{AA^\top - \Id_d} \le O\Paren{\delta + \sqrt{\e} + \sqrt{\e \lambda} + \e r}\,.
\]
Therefore, applying \cref{lem:stability-linear-transformation} again, we get the desired result.
\end{proof}

\begin{lemma}\label{lem:adequate}
For each $R > 0$, $\cP_R$ is adequate with respect to $\cS_R$  (see \cref{def:adequate}). 
\end{lemma}
\begin{proof}
Let $A \in \R^{d\times d}$. Then  
$\pE \iprod{A\otimes A, v^{\otimes 4}} = \pE \Paren{v^\top A v}^2 \ge 0$.

Let us prove the second property. Let $A, B\in \R^{d\times d}$. Then
\[
\Paren{\pE \iprod{A \otimes B, v^{\otimes 4}}}^2 = \Paren{\pE \Paren{v^\top A v}\Paren{v^\top B v}}^2 \le \pE \Paren{v^\top A v}^2 \cdot \pE  \Paren{v^\top B v}^2\,,
\]
where we used the \CS inequality for pseudo-distributions. 
By \cref{fact:spectral-certificates}, there is a degree-2 sos proof of the inequalities
\[
-\norm{A}\cdot \norm{v}^2\le {v^\top A v} \le \norm{A}\cdot \norm{v}^2\,.
\]
Hence each degree-4 pseudo-distribution that satisfies the constraint $\norm{v}^2\le R$,
\[
\Paren{\pE \iprod{A \otimes B, v^{\otimes 4}}}^2 \le \pE \Paren{v^\top A v}^2 \cdot \pE  \Paren{v^\top B v}^2\le R^2 \norm{A}^2 \cdot \norm{B}^2 = \Paren{\max_{uu^\top\in\cS_R}\iprod{uu^\top, A}^2} \cdot \Paren{\max_{uu^\top\in\cS_R}\iprod{uu^\top, B}^2}\,.
\]
\end{proof}

\cref{cor:spectral-stability}, \cref{lem:stability-linear-transformation} and \cref{lem:adequate} imply that we can use \cref{thm:filtering} and get the desired estimator in relative spectral norm.
\section{Estimation in Frobenius Norm}

\begin{proposition}\label{prop:Hanson-Wright}
    Let $0 < \e \lesssim 1$ and $R \le O(1)$.
    Let $\cD$ be a distribution in $\R^d$ with zero mean and identity covariance that satisfies $O(1)$-Hanson-Wright property, and let $x_1\ldots,x_n \simiid \cD$ for $n\gtrsim d^2\log^5(d)/\e^2$. 
    Then, with high probability, the set $S = \set{x_1x_1^\top, \ldots, x_nx_n^\top}$ is
 $\Paren{\e,\delta,r,\cB_R, \cB_R\otimes \cB_R}$-stable with respect to $M=\Id_d$ and $Q = \E \Paren{x_ix_i^\top - \Id}^{\otimes 2}$, where $\e \le \delta \le O(\e\log(1/\e))$ and $r\le O(1)$. 
\end{proposition}

\begin{proof}
    Condition 1 of stability follows from $O(1)$-sub-Gaussianity. Note that sub-Gaussianity follows from the Hanson-Wright if we plug $A = uu^\top$ for unit $u \in \R^d$. 
    Let us show condition 2.  
    Let $J\subseteq[n]$ be a set of size $\card{J} \ge (1-\e)n$,  and let $V\in \cB_R$. Note that
    \begin{align*}
    \tfrac{1}{\card{J}} \sum_{i\in J} \iprod{V^{\otimes 2}, \paren{x_ix_i^\top-\Id_d}^{\otimes 2}} 
    = \tfrac{1}{\card{J}} \sum_{x=1}^n \iprod{V^{\otimes 2}, \paren{x_ix_i^\top-\Id_d}^{\otimes 2}} 
    - \tfrac{1}{\card{J}} \sum_{i\in [n]\setminus J} \iprod{V^{\otimes 2}, \paren{x_ix_i^\top-\Id_d}^{\otimes 2}}\,.
    \end{align*}
    By the Hanson-Wright inequality, with high probability all $d^2$-dimensional vectors $x_ix_i^\top-\Id_d$ have norm at most $O(d)$. Hence we can apply matrix Bernstein inequality, and it implies that with high probability the sample covariance is close to the empirical covariance. That is, with high probability, for all $V\in \cB_R$,
    \[
    \Abs{\tfrac{1}{n}\sum_{x=1}^n \iprod{V^{\otimes 2}, \paren{x_ix_i^\top-\Id_d}^{\otimes 2}}
    - \E \iprod{V^{\otimes 2}, \paren{x_ix_i^\top-\Id_d}^{\otimes 2}}} \le O(\e)\,.
    \]

    Note that the Hanson-Wright inequality also implies that $\E \iprod{V^{\otimes 2}, \paren{x_ix_i^\top-\Id_d}^{\otimes 2}} \le O(1)$. Hence 
    \begin{align*}
    \tfrac{1}{\card{J}} \sum_{x=1}^n \iprod{V^{\otimes 2}, \paren{x_ix_i^\top-\Id_d}^{\otimes 2}} &= \Paren{1+O(\e)}\cdot \Paren{\E \iprod{V^{\otimes 2}, \paren{x_ix_i^\top-\Id_d}^{\otimes 2}} + O(\e)} 
    \\&= \E \iprod{V^{\otimes 2}, \paren{x_ix_i^\top-\Id_d}^{\otimes 2}} + O(\e)
    \\&= \iprod{V^{\otimes 2}, Q} + O(\e)\,.
    \end{align*}

    Bounding $\tfrac{1}{\card{J}} \sum_{i\in [n]\setminus J} \iprod{V^{\otimes 2}, \paren{x_ix_i^\top-\Id_d}^{\otimes 2}}$ by $O(\e \log^2(1/\e))$ can be done in exactly the same way as in Lemma 5.21 from \cite{DiakonikolasKK016}. For that Lemma they need condition 4 of Definition 5.15, that is shown in the proof of Lemma 5.17. Their proof uses only the Hanson-Wright inequality and no other properties of Gaussians, and it is enough to have $n\gtrsim d^2\log^5(d)/\e^2$ samples for it.

    As was observed before, the Hanson-Wright inequality implies condition 3 of stability: $\iprod{V^{\otimes 2}, Q} = \E \iprod{V^{\otimes 2}, \paren{x_ix_i^\top-\Id_d}^{\otimes 2}} \le O(1)$.
\end{proof}

\begin{lemma}\label{lem:frobenius-stability-isotropic}
  Let $0 < \e \lesssim 1$ and $R \le O(1)$.  Let $\Sigma \in \R^{d\times d}$ be a positive definite matrix such that $0.9 \cdot \Id_d \preceq \Sigma \preceq 1.1\cdot \Id_d$. Let $y_1,\ldots,y_n\simiid N(0,\Sigma)$, $\Sigma' = \Cov_{y\sim N(0,\Sigma)} \cF(y)$, where $\cF$ is as in \cref{def:function}. Let $\zeta_i = \Paren{\Sigma'}^{-1/2}\cF(y_i)$.

Then, with high probability, $\zeta_1\zeta_1^\top,\ldots,\zeta_n\zeta_n^\top$ is an $\Paren{\e,\delta,r,\cB_R, \cB_R\otimes \cB_R}$-stable set with respect to $M=\Id_d$ and $Q = \E_{y\sim N(0,\Sigma)} \Paren{\Paren{\Sigma'}^{-1/2}\cF(y)\cF(y)^\top\Paren{\Sigma'}^{-1/2} - \Id}^{\otimes 2}$, where $\e \le \delta \le O\paren{\e\log(1/\e)}$ and $r\le O(1)$. 
\end{lemma}
\begin{proof}
    By \cref{prop:Hanson-Wright} it is enough to show that the distribution of $\zeta$ satisfies $O(1)$-Hanson-Wright property. Since $\cF(y) = \cF(\Sigma^{1/2}g)$ is a $O(1)$-Lipschitz function of $g \sim \cN(0,1)$, and since $\Sigma'$ is $O(1/d)$-close to $\Sigma$ by \cref{lem:closeness}, $\Paren{\Sigma'}^{-1/2}\cF(\Sigma^{1/2}g)$ is also a $O(1)$-Lipschitz function of $g \sim \cN(0,1)$. By \cite{log-sobolev-are-hanson-wright}, this distribution satisfies $O(1)$-Hanson-Wright property.
\end{proof}

\begin{lemma}\label{lem:frobenius-stability}
      Let $0 < \e \lesssim 1$ such that $\e\log(1/\e) \gtrsim \log^2(d)/d$.  Let $\Sigma \in \R^{d\times d}$ be a positive definite matrix such that 
      \[
      \Paren{1-O(\e\log(1/\e)} \cdot \Id_d \preceq \Sigma \preceq\Paren{1+O(\e\log(1/\e)} \cdot \Id_d\,.
      \]
      Let $y_1,\ldots,y_n\simiid N(0,\Sigma)$ and $\cF$ be as in $\cref{def:function}$.

Then, with high probability, $\cF(y_1)\cF(y_1)^\top,\ldots,\cF(y_n)\cF(y_n)^\top$ is an $\Paren{\e,\delta,r,\cB_1, \cB_1\otimes \cB_1}$-stable set with respect to $M=\E\cF(y)\cF(y)^\top$ and $S = \E_{g\sim N(0,\Id)}\Paren{gg^\top / \norm{g}^2 - \Id}^{\otimes 2}$, where $\e \le \delta \le O\paren{\e\log(1/\e)}$ and $r\le O(1)$. 
\end{lemma}

\begin{proof}
    By \cref{lem:closeness},
     \[
      \Paren{1-O(\e\log(1/\e)} \cdot \Id_d \preceq \Sigma \preceq\Paren{1+O(\e\log(1/\e)} \cdot \Id_d\,.
      \]
    Denote $\zeta_i = \Paren{\Sigma'}^{-1/2}\cF(y_i)$. Let $J\subseteq[n]$ be a set of size $\card{J} \ge (1-\e)n$,  and let $V\in \cB_1$. Note that
    \[
    \tfrac{1}{\card{J}} \sum_{i\in J} \iprod{V, \cF(y_i)\cF(y_i)^\top-\E\cF(y)\cF(y)^\top} = \tfrac{1}{\card{J}} \sum_{i\in J} \iprod{\Paren{\Sigma'}^{-1/2}V\Paren{\Sigma'}^{-1/2}, \zeta_i \zeta_i^\top -\Id_d}
    \]
    By \cref{lem:frobenius-stability-isotropic}, $\zeta_1\zeta_1^\top,\ldots,\zeta_n\zeta_n^\top$ is $\Paren{\e,\delta,r,\cB_2, \cB_2\otimes \cB_2}$-stable
    with respect to $\Id_d$ and $Q = \E_{y\sim N(0,\Sigma)} \Paren{\Paren{\Sigma'}^{-1/2}\cF(y)\cF(y)^\top\Paren{\Sigma'}^{-1/2} - \Id}^{\otimes 2}$, hence the quantity above is bounded by $\delta = O(\e\log(1/\e))$, so condition 1 of stability is satisfied.

    Let us show the second condition. Note that
    \begin{align*}
   \tfrac{1}{\card{J}} \sum_{i\in J} \iprod{V^{\otimes 2}, \Paren{\cF(y_i)\cF(y_i)^\top-\E\cF(y)\cF(y)^\top}^{\otimes 2}} 
   &=
    \tfrac{1}{\card{J}} \sum_{i\in J} \iprod{V, \Paren{\cF(y_i)\cF(y_i)^\top-\E\cF(y)\cF(y)^\top}}^2 
    \\&=
        \tfrac{1}{\card{J}} \sum_{i\in J} \iprod{U, \Paren{\zeta_i \zeta_i^\top -\Id_d}}^2 \,,
        \end{align*}
    where $U = \Paren{\Sigma'}^{-1/2}V\Paren{\Sigma'}^{-1/2}$. By stability of $\zeta_1\zeta_1^\top,\ldots,\zeta_n\zeta_n^\top$, this quantity is $O(\delta^2/\e)$-close to $\iprod{U^{\otimes 2}, Q}$. Hence it is enough to show that
    \[
    \Abs{\iprod{U^{\otimes 2}, Q} - \iprod{V^{\otimes 2}, S}} \le O(\e\log^2(1/\e))\,.
    \]
    Since $\Paren{\Sigma'}^{-1/2} = \Id + A$, where $\Norm{A}\le O(\e\log(1/\e))$, 
    \[
    U = V + A V + VA + AVA\,.
    \]
    Since $\normf{AV}\le \norm{A}\cdot \normf{V} \le O(\e\log(1/\e))$,
    $U = V + \Delta$, where $\normf{\Delta} \le O(\e\log(1/\e))$.
    Hence $\iprod{\Delta^{\otimes 2}, Q} \le O(\e^2\log^2(1/\e))$. By the \CS inequality, 
    \[
    \Abs{\Iprod{V\otimes \Delta, Q}}\le \sqrt{\iprod{\Delta^{\otimes 2}, Q} \cdot \iprod{V^{\otimes 2}, Q}} \le  O(\e\log(1/\e))\,.
    \]
    Hence $\Abs{\iprod{U^{\otimes 2}, Q} - \iprod{V^{\otimes 2}, Q}} \le O(\e\log(1/\e))$, and it is enough to bound $\Iprod{V^{\otimes 2}, Q-S}$. By \cref{lem:tensor}, in the eigenbasis of $\Sigma$, for all $i,j\in [d]$, $Q_{ijij} = S_{ijij} + O(1/d)$, and $Q_{iijj} = S_{iijj} + O\Paren{\e\log(1/\e)/d}$. First let us bound the entries of the form $ijjij$:
    \[
    \sum_{i,j\in[d]}V_{ij}V_{ij}\Paren{Q_{ijij} - S_{ijij}} \le O(1/d)\cdot \sum_{i,j\in[d]}V_{ij}^2 \le O(1/d) \le O(\e\log(1/\e))\,.
    \]
    Now consider the entries of the form $iijj$. Since $\abs{\sum_{i=1}^d V_{ii}} \le \sqrt{d}$,
    \[
    \sum_{i,j\in[d]}V_{ii}V_{jj}\Paren{Q_{ijij} - S_{ijij}} \le O(\e\log(1/\e)/d)\cdot d \le O(\e\log(1/\e))\,.
    \]
    Since all other components of both $Q$ and $S$ are $0$, we get the desired bound.

    Condition 3 of the stability is clearly satisfied for $S$ (in particular, it follows from the proof above, since $Q$ is very close to $S$, and this condition is satisfied for $Q$).
\end{proof}

\cref{lem:frobenius-stability} and \cref{thm:filtering} imply that we get the desired estimator in relative Frobenius norm.
\section{Applications}

\subsection{Robust PCA}

In this subsection we briefly discuss how to derive \cref{cor:pca} from \cref{thm:main}. By Davis-Kahan inequality \cite{davis-kahan}, $\hat{v}$ is either ${O}(\e\log(1/\e)/\gamma)$-close to $v$ or to $-v$. Let us without loss of generality assume that it is close to $v$. Then
\[
\normf{\hat{v}\hat{v}^\top - vv} = \normf{\Paren{\hat{v} - v}\Paren{\hat{v} + v}^\top + \Paren{\hat{v} + v}\Paren{\hat{v} - v}^\top} \le O\Paren{\norm{\hat{v} - v}} \le {O}(\e\log(1/\e)/\gamma)\,.
\]

\subsection{Robust covariance estimation}

In this subsection we show how to estimate the scaling factors assuming the Hanson-Wright property or sub-exponentiality. Further in this section we denote by $\Sigma$ the covariance of $\cD$. First we assume that the mean of $\cD$ in $0$, and then we explain how the case $\mu\neq 0$ can be reduced to it.

By \cref{thm:main}, we can compute $\hat{\Sigma}$ such that
\[
\Normf{\paren{\hat{\Sigma}}^{-1/2}\Paren{\gamma\Sigma}\paren{\hat{\Sigma}}^{-1/2} - \Id} \le O\Paren{\e\log(1/\e)}\,,
\]
where $\gamma > 0$ is unknown. Our goal is to find a good estimator $\hat{\gamma}$ of $\gamma$, and then use $\hat{\Sigma}/\hat{\gamma}$ as an estimator for $\Sigma$.

Let us use fresh samples (as before, we can split the original sample into several parts, and work with a new part for a new estimation), and let us apply the transformation $\hat{\Sigma}^{-1/2}$ to the samples. The covariance of the resulting distribution is $\tilde{\Sigma} = \paren{\hat{\Sigma}}^{-1/2}\Sigma \paren{\hat{\Sigma}}^{-1/2}$, and it satisfies
\[
\normf{\gamma\tilde{\Sigma} - \Id} \le O\Paren{\e\log(1/\e)}\,.
\]

In particular, $\effrank(\tilde{\Sigma})\ge \Omega(d)$. Let us robustly estimate $\Tr(\tilde{\Sigma})$. 
By the Hanson-Wright inequality,
    \[
    \Pr_{x\sim\cD} \Paren{\Abs{\norm{x}^2 - \Tr \tilde{\Sigma}} \ge t} \le 2\exp\Paren{-\frac{1}{C_{HW}} \min\Set{\frac{t^2}{\normf{\tilde{\Sigma}}^2}, \frac{t}{\norm{\tilde{\Sigma}}}}} \le 2\exp\Paren{-\Omega\Paren{ \min\Set{\frac{d \cdot t^2}{\Tr(\tilde{\Sigma})^2}, \frac{\sqrt{d} \cdot t}{\Tr(\tilde{\Sigma})}}}}
    \,,
    \]
    where we used the fact that $\effrank(\tilde{\Sigma})\ge \Omega(d)$. Hence for all even $p\in\N$,
    \[
    \Paren{\E\Abs{\norm{x}^2 - \Tr \tilde{\Sigma}}^p}^{1/p} \le O\Paren{\Tr\paren{\tilde{\Sigma}}\cdot p/\sqrt{d}}\,.
    \]
    Now, applying one-dimensional truncated mean estimation (see, for example, Proposition 1.18 from \cite{DK_book}), we get an estimator $\hat{T}$ such that with high probability
    \[
    \Abs{\hat{T} - \Tr \tilde{\Sigma}} \le O\Paren{\Tr\paren{\tilde{\Sigma}} \cdot \e\log(1/\e) / \sqrt{d}}\,.
    \]

    Hence
    \[
    \Abs{1 - \Tr \paren{\tilde{\Sigma}}/\hat{T}} \le O\Paren{\e\log(1/\e) / \sqrt{d}}\,.
    \]

    Let $\hat{\gamma} = d/\hat{T}$, and $\hat\Sigma' = \hat{\Sigma}/\gamma = \paren{\hat{T}/d}\cdot \hat\Sigma$. 
    Then
    \begin{align*}
\Normf{\paren{\hat{\Sigma}'}^{-1/2}\Paren{\Sigma}\paren{\hat{\Sigma}'}^{-1/2} - \Id} 
&=  \Normf{\hat{\gamma}\tilde{\Sigma} - \Id} 
\\&= \Normf{\frac{\hat{\gamma}}{\gamma} \gamma\tilde{\Sigma} - \Id} 
\\&= \Normf{\Paren{1+O\Paren{\e\log(1/\e) / \sqrt{d}}} \gamma\tilde{\Sigma} - \Id}
\\&\le \normf{\gamma\tilde{\Sigma} - \Id} + O\Paren{\e\log(1/\e) / \sqrt{d}}\Normf{\gamma\tilde{\Sigma}}
\\&\le
 O\Paren{\e\log(1/\e)} + O\Paren{\e^2\log^2(1/\e)/\sqrt{d}}
 + O\Paren{\e\log(1/\e) / \sqrt{d}}\Normf{\Id_d}
 \\&\le
 O\Paren{\e\log(1/\e)}\,.
\end{align*}

Now let us prove the spectral norm bound assuming $O(1)$-sub-exponentiality. As before, we work with a transformed distribution with covariance $\tilde{\Sigma}$ such that for some $\gamma > 0$, $\Normf{\paren{\hat{\Sigma}}^{-1/2}\Paren{\gamma\Sigma}\paren{\hat{\Sigma}}^{-1/2} - \Id} \le O\Paren{\e\log(1/\e)}$.
Since $x = \xi \cdot \tilde{\Sigma}^{1/2} g/\norm{g}$, where $g\sim\cN(0,\Id)$,
for all $u\in \R^d$ and even $p\in \N$,
\[
\Paren{\E\xi^p\iprod{u, g/\norm{g}}^p}^{1/p} \le O\Paren{p\Paren{\E\xi^2\iprod{u, g/\norm{g} }^2 }^{1/2}}\,.
\]

Since $\xi$ and $g$ are independent, 
\[
\Paren{\E\xi^p}^{1/p} \le O\Paren{p\Paren{\E\xi^2}^{1/2} \cdot \frac{\Paren{\E\iprod{u, g/\norm{g}}^2 }^{1/2}}{\Paren{\E\iprod{u, g/\norm{g}}^p}^{1/p}} }\,.
\]
Note that since $g/\norm{g}$ and $g$ are independent, $\E\iprod{u, g}^p = \E\norm{g}^p \cdot \E\iprod{u, g/\norm{g}}^p$. Since $\norm{g}^2$ has chi-squared distribution, its $(p/2)$-th moment is $2^{p/2}\cdot \Gamma(d/2 + p/2) / \Gamma(d/2)$. Hence for all $p\le d$,
\[
\Paren{\E\norm{g}^p}^{1/p} = \Theta\Paren{\sqrt{d}}\,.
\]
Since for all $p$,
\[
\Paren{\E\iprod{u, g}^p}^{1/p} = \Theta(\sqrt{p} \norm{u})\,,
\]
we get for all $p\le d$
\[
\Paren{\E\xi^p}^{1/p} \le O\Paren{\sqrt{p}\Paren{\E\xi^2}^{1/2} }\,.
\]

Now consider 
\[
\E\norm{x}^p = \E\xi^p \cdot \E \norm{\tilde{\Sigma}^{1/2} g/\norm{g}}^p\,.
\]
Since $\norm{\tilde{\Sigma}^{1/2}} = \Theta(1)$, for all $p\le d$ we get 
\[
\Paren{\E\norm{x}^p}^{1/p} \le O\Paren{\sqrt{p}\Paren{\E\xi^2}^{1/2}}\le 
O\Paren{\sqrt{p}\E \norm{x}^2}^{1/2}
\,.
\]

Applying one-dimensional truncated mean estimation with $p = O(\log(1/\e))$, we get an estimator $\hat{T}$ such that with high probability
    \[
    \Abs{\hat{T} - \Tr \tilde{\Sigma}} \le O\Paren{\Tr\paren{\tilde{\Sigma}} \cdot \e\log(1/\e)}\,.
    \]
The rest of the proof is exactly the same as for the Hanson-Wright distributions (with relative spectral norm instead of relative Frobenius norm).

Now let us explain how to deal with the case when the mean of $\cD$ is nonzero. As for covariance estimation, we can use symmetrization $y = \paren{x-x'}/\sqrt{2}$. Denote $z = x - \mu$ and $z' = x' - \mu'$. We need to show the Hanson-Wright inequality for all positive definite matrices $A$ with $\effrank(A) \ge \Omega(d)$. Since the inequality is scale invariant, let us assume $\Tr(A) = d$. Then 
\[
y^\top Ay- \Tr(A) = \frac{1}{2}\Paren{z^\top A z - \Tr(A)} + \frac{1}{2}\Paren{(z')^\top Az' - \Tr(A)} + \frac{1}{2}z^\top A z' + \frac{1}{2}(z')^\top A z\,.
\]
The first two terms are bounded by $O\Paren{\sqrt{d\log(1/\delta)} + \log(1/\delta)}$ with probability $1-\delta$ (by the Hanson-Wright inequality). Let us bound the third term. Let us condition on $z'$ and treat it as a constant. By the Hanson-Wright inequality, with probability $1-\delta$,
\[
\Paren{z^\top A z'}^2 \le O\Paren{\sqrt{\norm{Az'}^2\log(1/\delta)} + \norm{Az'}^2\log(1/\delta)}\,.
\]

By the Hanson-Wright inequality for $z'$, $\norm{Az'}^2 \le O(\norm{z'}^2) \le  d + O\Paren{\sqrt{d\log(1/\delta)} + \log(1/\delta)}$ with probability $1-\delta$. Note that we only need to apply the Hanson-Wright inequality for $\delta$ such that $\log(1/\delta) \le O(\log(1/\e)) \le O(\log d)$. For such $\delta$, 
\[
z^\top A z' \le O\Paren{\sqrt{d\log(1/\delta)}}.
\]
The term $(z')^\top A z$ is also bounded by this value. Since this bound is enough to obtain the bound on the $O(\log(1/\e)$-th moment that we use, we get the desired result.

Let us now consider the sub-exponential case. Note that
\[
O\Paren{\E \iprod{z, u}^p}^{1/p} \le O\Paren{\Paren{\E \iprod{x - \mu, u}^p}^{1/p} + \Paren{\E \iprod{x' - \mu, u}^p}^{1/p}} \le 
O\Paren{\Paren{\E \iprod{x - \mu, u}^2}^{1/2}}\,.
\]
Note that $\iprod{z, u}^2 = \tfrac{1}{2}\iprod{x - \mu, u}^2 + \tfrac{1}{2}\iprod{x' - \mu, u}^2 - \iprod{x - \mu, x'-\mu}$, so $\E \iprod{z, u}^2 = \E  \iprod{x - \mu, u}^2$, and we get the desired bound.


\section{Sum-of-Squares Proofs}
\label{sec:sos}

We use the standard sum-of-squares machinery, used in numerous prior works on algorithmic statistics, e.g. \cite{KS17, hopkins2018mixture,Bakshi, sos-poincare, sos-subgaussian}. 

In particular, we use the following facts:

\begin{fact}[Spectral Certificates] \label{fact:spectral-certificates}
For any $m \times m$ matrix $A$, 
\[
\sststile{2}{u} \Set{ \iprod{u,Au} \leq \Norm{A} \Norm{u}_2^2}\mper
\]
\end{fact}

\begin{fact}[Cauchy-Schwarz for Pseudo-distributions]
Let $f,g$ be polynomials of degree at most $d$ in indeterminate $x \in \R^d$. Then, for any degree d pseudo-distribution $D$,
$\pE_{D}[fg] \leq \sqrt{\pE_{D}[f^2]} \sqrt{\pE_{D}[g^2]}$.
 \label{fact:pseudo-expectation-cauchy-schwarz}
\end{fact}


\end{document}